\pdfoutput=1
\documentclass[reqno,centertags,11pt,a4paper]{amsart}
\usepackage[utf8]{inputenc}
\usepackage[T1]{fontenc}
\usepackage[english]{babel}

\usepackage[svgnames]{xcolor}
\usepackage[colorlinks=true,linkcolor=Green,citecolor=Orange,urlcolor=Blue,breaklinks=true]{hyperref}

\usepackage{amsfonts,amsmath,amssymb,amsthm,mathtools}
\usepackage{mathtools}
\usepackage{enumerate,multirow}
\usepackage[normalem]{ulem}

\newcommand{\N}{{\mathbb{N}}}

\newcommand{\R}{{\mathbb{R}}}

\newcommand{\C}{{\mathbb{C}}}

\newcommand{\ol}{\overline}

\newcommand{\wti}{\widetilde  }
\newcommand{\Oh}{O}

\newcommand{\loc}{\text{\rm{loc}}}

\newcommand{\beq}{\begin{equation}}
\newcommand{\eeq}{\end{equation}}
\newcommand{\bdm}{\begin{displaymath}}
\newcommand{\edm}{\end{displaymath}}
\newcommand{\ba}{\begin{align}}
\newcommand{\ea}{\end{align}}
\newcommand{\bpf}{\begin{proof}}
\newcommand{\epf}{\end{proof}}

\newcommand{\la}{\langle}
\newcommand{\ra}{\rangle}

\newcommand{\esssup}{\mathop{\mathrm{ess~sup}}}

\renewcommand{\d}{\mathrm{d}}
\newcommand{\veps}{\varepsilon}
\newcommand{\re}{\mathrm{Re}}

\newcommand{\tr}{\mathop{\mathrm{tr}} \nolimits}

\newcommand{\id}{\mathbf{1}}                %
\newcommand{\idG}{\mathbf{1}_{\calG}}       %

\allowdisplaybreaks

\newcommand{\calB}{\mathcal{B}}

\newcommand{\calF}{\mathcal{F}}
\newcommand{\calG}{\mathcal{G}}
\newcommand{\calH}{\mathcal{H}}

\newcommand{\calS}{{\mathcal{S}}}

\newtheorem{theorem}{Theorem}[section]
\newtheorem{proposition}[theorem]{Proposition}
\newtheorem{lemma}[theorem]{Lemma}
\newtheorem{corollary}[theorem]{Corollary}

\theoremstyle{definition}

\newtheorem{remark}[theorem]{Remark}
\newtheorem{remarks}[theorem]{Remarks}

\newcounter{theoremi}[theorem]

\newcommand{\itemthm}{\refstepcounter{theoremi} {\rm(\roman{theoremi})}{~}}

\numberwithin{theorem}{section}
\numberwithin{equation}{section}
\begin{document}
\title[Cwikel's bound reloaded]{Cwikel's bound reloaded}

\author[D.~Hundertmark]{Dirk Hundertmark}
\address{Department of Mathematics, Institute for Analysis, Karlsruhe Institute of Technology, 76128 Karlsruhe, Germany.}
\address{Department of Mathematics, Altgeld Hall, University of Illinois at Urbana-Champaign, 1409 W. Green Street, Urbana, IL 61801}
\email{dirk.hundertmark@kit.edu}

\author[P.~Kunstmann]{Peer Kunstmann}
\author[T.~Ried]{Tobias Ried}
\author[S.~Vugalter]{Semjon Vugalter}
\address{Department of Mathematics, Institute for Analysis, Karlsruhe Institute of Technology, 76128 Karlsruhe, Germany.}
\email{peer.kunstmann@kit.edu}
\email{tobias.ried@kit.edu}
\email{semjon.wugalter@kit.edu}

\thanks{\copyright 2018 by the authors. Faithful reproduction of this article, in its entirety, by any means is permitted for non-commercial purposes}
\date{\today, version \jobname}


\begin{abstract} 

There are a couple of proofs by now for the famous Cwikel--Lieb--Rozenblum (CLR) bound, which is a semiclassical bound on the number of bound states for a Schr\"odinger operator, proven in the 1970s.
Of the rather distinct proofs by Cwikel, Lieb, and Rozenblum, the one by Lieb gives the best constant, the one by Rozenblum does not seem to yield any reasonable estimate for the constants, and Cwikel's proof is said to give a constant which is at least about 2 orders of magnitude off the truth. This situation did not change much during the last 40+ years. 

It turns out that this common belief, i.e, Cwikel's approach yields bad constants, is not set in stone: We give a drastic simplification of Cwikel's original approach which leads to an astonishingly good bound for the constant in the CLR inequality. Our proof is also quite flexible and leads to rather precise bounds for a large class of Schr\"odinger-type operators with generalized kinetic energies. Moreover, it  highlights a natural but overlooked connection of the CLR bound with bounds for maximal Fourier multipliers from harmonic analysis.
\end{abstract}

\maketitle
{\hypersetup{linkcolor=black}
\tableofcontents}

\section{Introduction}\label{introduction}
We want to find natural bounds, with the right semi-classical behavior, for the number of negative eigenvalues of Schr\"odinger  operators $P^2+V$ with $P=-i\nabla$, the momentum operator, or more general operators like the polyharmonic Schr\"odinger operators $|P|^{2\alpha}+V$, including the ultra-relativistic operator $|P|+V$. We will also consider operator-valued potentials $V$.  

For the one-particle Schr\"odinger operator $P^2+V$ with $P=-i\nabla$ the momentum operator and $V$ a real-valued potential, this type of bound goes back to Cwikel, Lieb, and Rozenblum \cite{Cwikel,Lieb,Ro1,Ro2}, with very different proofs. They prove 
\begin{align}\label{eq:original-CLR}
	N(P^2+V) \le L_{0,d}  \int_{\R^d} V_-(x)^{d/2}\, dx
\end{align}
for the number of negative eigenvalues of a Schr\"odinger operator.    
This bound is a semi-classical bound since a simple scaling argument shows that the classical phase-space volume of the region of negative energy is given by
\begin{align}
	N^\text{cl}(\eta^2+V)= \iint_{\eta^2+V(x)<0} 1\,\frac{d\eta\, dx}{(2\pi)^d} =  \frac{|B_1^d|}{(2\pi)^d}\int_{\R^d} V_-(x)^{d/2}\, dx \, .
\end{align} 
where $|B_1^d|$ is the volume of the unit ball in $\R^d$. 

The intuition is that the uncertainty principle forces a quantum particle to occupy roughly a classical phase-space 
volume $(2\pi)^d$. Thus  $N^\text{cl}(\eta^2+V)$, counting the 
volume where  the classical Hamiltonian energy 
$H(\eta,x) = \eta^2+V(x)$ is negative, should control $N(P^2+V)$.  
The CLR bound  \eqref{eq:original-CLR} shows that this is the case modulo the factor\footnote{We write $L_{0,d}$ etc, since 
there are a class of inequalities for the $\gamma^\text{th}$ moment 
of the negative eigenvalues with associated constants 
$L_{\gamma,d}$, see \cite{LT1,LT2} and the reviews \cite{LaWe-review,Hu-review}.} $C_{0,d}=L_{0,d}\frac{(2\pi)^d}{|B_1^d|}$ . 

The original bounds on $C_{0,d}$ in \cite{Cwikel} and \cite{Lieb} 
were explicitly dimension dependent with a considerable growth in the dimension $d$. The bound due to Lieb grows like $C_{0,d}= \sqrt{\pi d}(1+\Oh(d^{-1}))$. 
See \cite{Simon} or \cite[Chapter 3.4]{Roepstorff} for an excellent discussion of Lieb's method and Remark \ref{rem:explicit-numbers} below for some explicit numbers. 
However, it is expected that semi-classical arguments work better in high  dimensions. In particular, the constant $C_{0,d}$ should not grow in $d$. 
The first dimension independent bound  $C_{0,d}\le 81$ was derived  
by extending Cwikel's method to operator-valued potentials in 2002 in \cite{Hu}. This work extended an induction in the dimension argument by Laptev and Weidl\footnote{See also \cite{La} for some indication of the induction in dimension trick.} \cite{LaWe}, who were the first to derive Lieb--Thirring bounds with the sharp classical Lieb--Thirring constant in all dimensions in some cases. 
Although the upper bound from \cite{Hu} is dimension independent, it is certainly too large for small dimensions. 

For the last $40$-plus years it has been believed that any approach 
based on Cwikel's method cannot, in low dimensions, yield any bounds on $C_{0,d}$ which are comparable to the ones obtained by Lieb. 
This is wrong, as we will show by drastically simplifying and, at the same time, generalizing the important ideas of Cwikel. 
A typical result which can be easily achieved with our method is 
\begin{theorem}\label{thm:CLR}
	The number $N(P^2+V)$ of negative energy bound states of  $P^2+V$ obeys the semiclassical bound  
	\begin{align}
		N(P^2+V)\le C_{0,d} \frac{|B^d|}{(2\pi)^d} \int_{\R^d} V_-(x)^{d/2}\, \d x
	\end{align}
	for all $d\ge 3$, where $B^d$ is the unit ball in $\R^d$, $|B^d|$ its volume, and the constant $C_{0,d}$ given in Table \ref{tb:better} below. 
	
	Moreover, the same bounds with the same constants also hold in the operator-valued case, see Theorem \ref{thm:great-operator-valued}.
\end{theorem}
\begin{remarks}\label{rem:explicit-numbers}  
\itemthm Table \ref{tb:better} below compares the upper bounds on $C_{0,d}$,  obtained with our method, with the best known ones so far for scalar and operator-valued potentials,  
  \begin{table}[ht!]
  \begin{tabular}{| l | c | c | c|}
    \hline
    \multirow{2}{*}{$d$} & Our results  & \multicolumn{2}{c|}{Best known so far} \\
        & scalar and operator valued  & scalar & operator-valued         \\\hline
    3   & 7.55151 & 6.86924  &  \multirow{7}{*}{10.332}\\ 
    4   & 6.32791 & 6.03398  &  \\  
    5   & {5.95405} & 5.96677 &   \\   
    6   & {5.77058} & 6.07489 &  \\ 
    7   & {5.67647} & 6.24464 &  \\ 
    8   & {5.63198} & 6.43921 &  \\ 
    9   & {5.62080}  & 6.64378 &  \\ \hline
  \end{tabular}
  \vskip1ex
  \caption{}
  \label{tb:better}
  \end{table}
  All bounds on $C_{0,d}$ in the third column of the table were obtained already in the original work of Lieb more than 40 years ago\footnote{The numbers are taken from Roepstorff's book \cite[Table 3.1]{Roepstorff}}. 
 Our bounds on $C_{0,d}$ also hold in the operator-valued case, see Section \ref{sec:operator-valued} below. The last column is due to Frank, Lieb and 
 Seiringer \cite{FLS} and holds for all $d\ge 3$. 
 Our result also gives the bound $C_{0,d}\le 5.62080 $ for $d\ge 9$, see the discussion in Appendix \ref{sec:induction-in-dimension}.  For dimensions $d=3,\ldots,9$ our upper bounds are compared with the values of the lower bound \eqref{eq:CLR-lower-bound} achievable by our method in Table \ref{tb:lower-bound} below. \\[0.2em] 
\itemthm There have been several previous attempts to improve on Lieb's result, for example, due to Conlon \cite{Co}, Li and Yau \cite{LY}, Frank \cite{Fr}, and Weidl \cite{Weidl1,Weidl2}. All these very much different proofs shed a new light on the  Cwikel--Lieb--Rozenblum bound, but failed to give better bounds on the involved constants than already achieved by Lieb. \\[0.2em]
\itemthm From the point of view of physics, the other important case is $\alpha=1/2$, which corresponds to an ultra-relativistic Schr\"odinger operator $|P|+V$. In three dimensions we get the upper bound
  \begin{align}
  	N(|P|+V) \le 5.77058 \int_{\R^3} V_-(x)^{3}\, dx
  \end{align}
  which improves the result of Daubechies \cite{Daubechies}, who gets $N(|P|+V) \le 6.08 \int_{\R^3} V_-(x)^{3}\, dx$.
 \end{remarks}

For more general so-called polyharmonic Schr\"odinger-type operators our method yields the following for scalar potentials. 
A similar result, with the same constants, also holds for operator-valued potentials, see Theorem \ref{thm:poly-operator-valued}.  

\begin{theorem}\label{thm:poly}
	Let $P=-i\nabla$ be the momentum operator, $V=V_+-V_-$ be a real-valued potential with positive part $V_+\in L^1_{\loc}$ and negative part 
	$V_-\in L^{d/\alpha}(\R^d)$ with $0<\alpha<d/2$, and $P^{2\alpha}+V$ the Schr\"odinger--type operator defined via quadratic form methods on $L^2(\R^d)$. Then the number $N(P^{2\alpha}+V)$ of negative energy bound states of  $P^{2\alpha}+V$ is bounded by 
	\begin{align}\label{eq:poly}
		N(P^{2\alpha}+V)
			\le C_{d/\alpha}\, 
			\frac{|B_1^d|}{(2\pi)^d} \int_{\R^d} V_-(x)^{\frac{d}{2\alpha}}\, dx
	\end{align}
	with constant 
	\begin{align}\label{eq:C-gamma}
		C_{\gamma} = \frac{\gamma^{\gamma +1}}{4\left( \gamma-2 \right)^{\gamma-2}} 
						   M_\gamma \ ,
	\end{align}
  where 
  \begin{align}\label{eq:M-gamma}
  	M_\gamma= \inf\left\{  
  						\left( \|m_1\|_{L^2(\R_+,\frac{ds}{s})} 
  								\|m_1\|_{L^2(\R_+,\frac{ds}{s})} 
  						\right)^{\gamma-2} 
  							\int_0^\infty (1-t^{-1}m(t))^2 \, t^{1-\gamma}\, dt \, 
  				  \right\},
  \end{align}
  the infimum taken over all $m_1,m_2\in L^2(\R_+,\frac{ds}{s})$, and $m= m_1*m_2$ denotes the convolution of $m_1, m_2$ on $\R^+$ with measure $ds/s$. 
\end{theorem}
The minimization problem for $M_\gamma$ in \eqref{eq:M-gamma} seems to be new. As Theorem \ref{thm:poly}  shows, it has considerable implications for the spectral theory of Schr\"odinger operators. 
For the constant $M_\gamma$ above, we note the following estimate. \begin{proposition}\label{prop:simpel-bound-M-gamma}
	For all $\gamma>2$ 
	\begin{align}\label{eq:easy-estimates}
		\frac{2}{\gamma(\gamma-1)(\gamma-2)}
		\le M_\gamma 
		\le \frac{8 }{\gamma(\gamma-2)(\gamma+ 2) }
	\end{align}
\end{proposition}
For the proof see Sections \ref{sec:lower-bound} and Remark \ref{rem:simple-choice} below.
 
\begin{remarks}
\itemthm
	Even this simple upper bound on $M_\gamma$ yields result better than available results so far in the literature: Using ideas from Rumin \cite{Ru1, Ru2}, Frank \cite{Fr} got the bound 
	\begin{align*}
		N(P^{2\alpha}+V)
			\le \left(
					\frac{d(d +2\alpha)}{(d-2\alpha)^2}
				\right)^{(d-2\alpha)/(2\alpha)} \frac{d}{d -2\alpha}
			\frac{|B_1^d|}{(2\pi)^d} \int_{\R^d} V_-(x)^{\frac{d}{2\alpha}}\, dx\, .
	\end{align*}
   Computing the ratio of the constants in Frank's bound and the one from \eqref{eq:poly}, using the upper bound in \eqref{eq:easy-estimates}, one sees that our bound from Theorem 
   \ref{thm:poly}  is better in the whole allowed range of $0<\alpha<d/2$. \\[0.2em]
\itemthm For the constant $C_\gamma$ in \eqref{eq:C-gamma},  the lower bound from \eqref{eq:easy-estimates} yields  
   \begin{align*}
 	  C_\gamma \ge \frac{\gamma^{\gamma}}{2(\gamma-1)\left( \gamma-2 \right)^{\gamma-1}} \eqqcolon C^\text{lower}_\gamma\, ,
   \end{align*}
  where $C^\text{lower}_\gamma$ is a, probably non-sharp, lower bound for the best possible constant achievable by our method\footnote{Which is not necessarily the best possible constant.}. 
   Thus the upper bound on $C_\gamma$ shows 
  \begin{align*}
  	\frac{C_\gamma}{C^\text{lower}_\gamma}\le  4\frac{\gamma-1}{\gamma+2}<4\, ,
  \end{align*}
  where $\gamma=d/\alpha> 2$. 
  So our upper bound is less than a factor of 4 off the lower bound.  \\
\itemthm The above lower bound also gives the lower bound 
\begin{align}\label{eq:CLR-lower-bound}
	C^\text{lower}_{0,d} = C^\text{lower}_d =  \frac{d^{d}}{2(d-1)\left( d-2 \right)^{d-1}}
\end{align}
for the constant in Theorem \ref{thm:CLR}, achievable by our method. For dimensions $3\le d \le 9$ this gives the lower bounds 
  \begin{table}[h!]
  \begin{tabular}{| l | c | c |}
    \hline $d$ & Our results  & lower bound    \\\hline
    3   & 7.55151 & 6.75000  \\ 
    4   & 6.32791 & 5.33333  \\  
    5   & {5.95405} & 4.82253  \\   
    6   & {5.77058} & 4.55625   \\ 
    7   & {5.67647} & 4.39229   \\ 
    8   & {5.63198} & 4.28088   \\ 
    9   & {5.62080}  & 4.20028   \\ \hline
  \end{tabular}
  \vskip1ex
  \caption{}
  \label{tb:lower-bound}
  \end{table}
  
 In addition, 
 \begin{align*}
 	C^\text{lower}_{0,d} = \frac{d^2}{2(d-1)(d-2)}\left( 1+\frac{2}{d-2} \right)^{d-2} \rightarrow \frac{e^2}{2}\ge 3.69452\, .
 \end{align*}
 This comparison shows that there is not too much room to improve 
 on the upper bounds we obtained, even if one finds the sharp value in the minimization problem for $M_\gamma$ in \eqref{eq:M-gamma}.

\itemthm It is known that if $\alpha\ge d/2$, the operator $P^{2\alpha}-U$ always has bound states for nontrivial $U\ge 0$, so a quantitative bound of the form $N(P^{2\alpha}-U)\lesssim \int_{\R^d}U(x)^{d/\alpha}$ cannot hold if $\alpha\ge d$. 
For $\alpha=1$ see \cite{Simon2} or \cite[Problem 2 in \S 45]{LL}. 
For more general cases, see \cite{NW96,LSW,Pankrashkin}, and \cite{HHRW} for a simple proof of how the existence/ non-existence of a CLR type bound for operators of the form $T(P)+V$ for a large class of functions $T:\R^d\to [0,\infty)$ is related to the behavior of the symbol $T$ close to its zero-set. 
\end{remarks}

As we mentioned before, our method can be generalized to  operator valued potentials, leading to the same results. 
To formulate this, we need some additional notation. 
An operator-valued potential $V$ is a map $V:\R^d\mapsto \calG$ with $V(x):\calG\to \calG$ a bounded self-adjoint operator on an auxiliary Hilbert space\footnote{In the following, all Hilbert-spaces are considered to be separable, unless said otherwise ;-).  Physically, this auxiliary Hilbert-space corresponds to other of degrees of freedom, for example spin.} $\calG$ for almost all $x\in\R^d$. We denote by $\calB(\calG)$ the bounded operators on $\calG$ and by $\calS_\alpha(\calG)$ the von Neumann--Schatten ideal of compact operators on $\calG$ with $\alpha$ summable singular values, see for example \cite{Simon-trace-ideals} for a background on  von Neumann--Schatten ideals.

\begin{theorem}[Operator-valued version of Theorem \ref{thm:poly}]\label{thm:poly-operator-valued}
 Let $\calG$ be a Hilbert space  and $V:\R^d\to \calB(\calG)$ an operator valued potential with positive part $V_+\in L^1_{\text{loc}}(\R^d,\calB(\calG))$ and negative part  $V_-\in L^{d/(2\alpha)}(\R^d, \calS_{d/(2\alpha)}(\calG))$. Then the number of negative energy bound states of $P^{2\alpha}\otimes\idG+V$ is bounded by 
 \begin{align}
 	N(P^{2\alpha}\otimes\idG+V)
			\le C_{d/\alpha}\, 
			\frac{|B_1^d|}{(2\pi)^d} \int_{\R^d} \tr_\calG[ V_-(x)^{\frac{d}{2\alpha}}]\, dx
 \end{align}
 with the same constant $C_\gamma$ as in Theorem \ref{thm:poly}. 
\end{theorem}

For the physically most interesting case $\alpha=1$ this enables us to get  considerable  improvements on the constants in the Cwikel--Lieb--Rozenblum bound. 

\begin{theorem}[Operator-valued version of Theorem \ref{thm:CLR}]\label{thm:CLR-op-valued}
	Let $\calG$ be a Hilbert space  and $V:\R^d\to \calB(\calG)$ an operator valued potential with positive part $V_+\in L^1_{\text{loc}}(\R^d,\calB(\calG))$ and negative part  $V_-\in L^{d/2}(\R^d, \calS_{d/2}(\calG))$. Then the number of negative energy bound states of $P^{2}\otimes\idG+V$ is bounded by 
 \begin{align}
 	N(P^{2}\otimes\idG+V)
			\le C_{0,d}^{\mathrm{op}}\, 
			\frac{|B_1^d|}{(2\pi)^d} \int_{\R^d} \tr_\calG[ V_-(x)^{\frac{d}{2}}]\, dx
 \end{align}
 with 
 \begin{align}\label{eq:constants-op-valued}
 	C_{0,d}^{\mathrm{op}} = \min_{3\leq n \leq d} C_{0,n}^{\mathrm{op}} \leq \min_{3\leq n \leq d} C_n,
 \end{align} 
 where $C_n$ is given by \eqref{eq:C-gamma} for $\gamma=n$.
\end{theorem}

\begin{remark}
	Table \ref{tb:better} lists	upper bounds on $C_{0,d}^{\mathrm{op}}$ for dimensions $d=3\ldots 9$, see also Appendix \ref{sec:numerics}. The constant for $d=9$ is also an upper bound on $C_{0,d}^{\mathrm{op}}$ in any dimension $d\geq 10$ by \eqref{eq:constants-op-valued}.
\end{remark}

The structure of the paper is as follows. In Section \ref{sec:CLR} we present the main ideas of our method in the case of a standard non-relativistic Schr\"odinger operator. The extension to more general kinetic energies is done in Section \ref{sec:gen-kin-energy}.

In Section \ref{sec:max-F-mult} we explain the surprising connection of semiclassical bounds and maximal Fourier 
multiplier estimates, which is probably the most important new part 
of our method.

Although we cannot explicitly find minimizers of the variational problem from Theorem \ref{thm:poly}, there is a natural lower bound, which is discussed in Section \ref{sec:lower-bound}. The numerical study to find reasonable  upper bounds for this variational problem is presented in Appendix \ref{sec:numerics}. 

The extension to the operator-valued setting is done in Sections \ref{sec:operator-valued} and \ref{sec:trace-ideals}. 
In particular, in Section \ref{sec:trace-ideals} we prove a 
fully operator-valued version of Cwikel's original weak trace ideal bound.

\section{The splitting trick}\label{sec:CLR}
The main idea in the proof of Theorem \ref{thm:CLR} is quite simple. 
Let $U\coloneqq V_-\ge 0$. As quadratic forms $P^2+V\ge P^2-U$.
This and the  Birman--Schwinger principle show that
\begin{align*}
	N(P^2+V)\le N(P^2-U) = n(U^{1/2}|P|^{-2}U^{1/2};1) 
\end{align*}
where $N(P^2+V)$ are the number of negative eigenvalues of $P^2-U$ and $n(A;\kappa)$ the number of singular values $(s_j(A))_{j\in\N }$ greater than $\kappa>0$ of a compact operator $A$.

We denote by $\calF$ the Fourier transform and by  $\calF^{-1}$ its inverse, by $M_h$ the operator of  multiplication with a function $h$,    
and  $A=A_{f,g} = M_f\calF^{-1}M_g$ for $f,g$ non-negative (measurable) functions on $\R^d$. When $f(x)=U(x)^{1/2}$ and $g(\eta)= |\eta|^{-1/2}$ then $A A^* = U^{1/2}|P|^{-2}U^{1/2}$, which has the same non-zero eigenvalues as $A^*A$. Thus 
\begin{align*}
	N(P^2-U) = n(A_{f,g};1) 
\end{align*}
In particular, the Chebyshev--Markov inequality gives 
\begin{align*}
	N(P^2-U)
	  & = n(A_{f,g};1) 
	  	\le \sum_{j} \frac{(s_j(A_{f,g})-\mu)_+^2}{(1-\mu)^2}
\end{align*}
for any $0<\mu<1$. 
Let's drop the dependence of $A$ on $f$ and $g$ for the moment.
We want to split $A= B+ H$, where $B$ is bounded,  and note that Ky Fan's inequality for the singular values  yields 
\begin{align*}
	s_j(A) &= s_j(B+ H) 
	  \le s_1(B) + s_j(H) \\
	&\le \|B\| + s_j(H)
\end{align*}
for all $j\in\N $. So if $\|B_{f,g}\|\le \mu<1$ we get  
\begin{align}\label{eq:first-upper-bound}
	N(P^2-U) \le (1-\mu)^{-2} \sum_{j\in\N} s_j(H_{f,g})^2
	  = (1-\mu)^{-2} \|H_{f,g}\|_{HS}^2,
\end{align}
where $\|H\|_{HS}$ denotes the Hilbert--Schmidt norm of an operator $H$. 

In order to make the above argument work, one has to be able to split $A_{f,g}= B_{f,g}+H_{f,g}$ in such a way that the Hilbert-Schmidt norm of $H_{f,g}$ is easy to calculate and one has a good bound on the operator norm of $B_{f,g}$. 
It will turn out, see the calculation below, that $\|H_{f,g}\|_{HS}^2= c \int_{R^d}f(x)^d\, dx$, so the right hand side of \eqref{eq:first-upper-bound} has exactly the right (semi-classical) scaling in $f$. 
But, in order to use this in \eqref{eq:first-upper-bound}, it also enforces that the upper bound $\mu$ on the operator norm of $B_{f,g}$ has to be \emph{independent} of $f$. Since for a given $\varphi\in L^2$ one can freely choose $f\ge 0$ as to make $|B_{f,g,m}\varphi|$ as big as possible, this leads naturally to the associated maximal operator $B^*_{g,m}\coloneqq\sup_{f\ge 0}|B_{f,g,m}\varphi|$.  
Although this is not explicitly written in the paper by Cwikel, getting a useful bound on such a type of maximal operator is exactly what he achieved in \cite{Cwikel}, using  a dyadic decomposition in the ranges of $f$ and $g$ and collecting suitable terms. 
We will do this in a much simpler and more effective way. This enables us to get a constant which is more than 10 times smaller than the original constant by Cwikel.

Writing out the Fourier transform, one sees that  $A_{f,g}$ has kernel 
\begin{align}
	A_{f,g}(x,\eta)= (2\pi)^{-d/2} e^{ix\cdot\eta} f(x)g(\eta),
\end{align} 
that is,
\begin{align}
	A_{f,g}\varphi(x)= f(x)\calF^{-1}(g\varphi)(x)
	= (2\pi)^{-d/2}  \int_{\R^d} e ^{ix\cdot\eta} f(x)g(\eta)\varphi(\eta)\, \d\eta ,
\end{align}
at least for nice enough $\varphi$. 

In order to write $A_{f,g}$ as a sum of a bounded and a Hilbert-Schmidt operator, set $t=f(x) g(\eta)$, split $t= m(t) +t-m(t)$ for some function $m:[0,\infty)\to \R $, and define $B=B_{f,g,m}$ and $H_{f,g,m}$ via their kernels
\begin{align}\label{eq:kernel-B}
	B_{f,g,m}(x,\eta) &= (2\pi)^{-d/2} e^{ix\cdot\eta}m(f(x)g(\eta)) \\
	H_{f,g,m}(x,\eta) &=  (2\pi)^{-d/2} e^{ix\cdot\eta} \left(f(x)g(\eta)-m(f(x)g(\eta))\right) \label{eq:kernel-H}
\end{align}
It is then clear that $A_{f,g}= B_{f,g,m}+H_{f,g,m}$. Our starting point is that the Hilbert--Schmidt norm of $H_{f,g,m}$ is easy to calculate and it is not too hard to get an explicit bound on the operator norm of $B_{f,g,m}$ on $L^2$ under a suitable assumption on $m$. 
\begin{theorem}\label{thm:great}  
	The Hilbert--Schmidt norm of $H_{f,g,m}$ is given by 
	\begin{align}
		\|H_{f,g,m}\|_{HS}^2= \int_{\R^d} G_{g,m}(f(x)) \, dx
	\end{align}
	where $G_{g,m}$ is given by 
	\begin{align}
		G_{g,m}(u)= \int_{\R^d} |ug(\eta)-m(ug(\eta))|^2 \frac{\d \eta}{(2\pi)^d} .
	\end{align}
    If, moreover, $m$ is given by a convolution,  that is,  
    \begin{align}
    	m(t) = m_1*m_2(t)=\int_0^\infty m_1(t/s)m_2(s)\frac{\d s}{s}
    \end{align}
    then for all measurable non-negative functions $f$ and $g$ the operator $B_{f,g,m}$ is bounded on $L^2(\R^d)$ with 
    \begin{align}\label{eq:B-bound}
    	\|B_{f,g,m}\varphi\|_{2} \le \left( \int_0^\infty |m_1(s)|^{2} \frac{d s}{s}\right)^{1/2}\left( \int_0^\infty |m_2(s)|^{2} \frac{d s}{s}\right)^{1/2}\|\varphi\|_2
    \end{align}
    for all $\varphi\in L^2(\R^d)$. 
\end{theorem}
\begin{remark}
	 We stress the fact, that the bound on the operator norm of $B_{f,g,m}$ is \emph{independent} of the choice of $f$, as it has to be, and also of $g$. This will turn out to be a natural consequence of 
	 the convolution structure of $m$, 
	 see Section \ref{sec:max-F-mult} below.
\end{remark}
\begin{proof}[Proof of the first half of Theorem \ref{thm:great}: Evaluating the Hilbert--Schmidt norm] 
 Since the operator $H_{f,g,m}$ has a kernel given by the right-hand side of \eqref{eq:kernel-H}, one computes its Hilbert-Schmidt norm as 
  \begin{align*}
  	\|H_{f,g,m}\|_{HS}^2 
  		&= \iint_{\R^d\times\R^d} |H_{f,g,m}(x,\eta)|^2  \d x d\eta 
  			= \iint_{\R^d\times\R^d} \left|f(x)g(\eta)-m(f(x)g(\eta))\right|^2  \, \frac{\d x d\eta}{(2\pi)^d} \\
  		&= \int_{\R^d} G_{g,m}(f(x))\, dx
  \end{align*}
  using the Fubini-Tonelli Theorem and the definition of $G_{g,m}$. 
  The proof of the second half of Theorem \ref{thm:great} will be given later in Section \ref{sec:max-F-mult} below. 
\end{proof}

In the rest of this section we will discuss how Theorem \ref{thm:great} and the bound \eqref{eq:first-upper-bound} easily lead to the Cwikel--Lieb--Rozenblum bound for a non-relativistic single-particle Schr\"odinger operator:  
In this case  
$g(\eta)=|\eta|^{-1}$. A simple scaling argument gives
\begin{align*}
	\|H_{f,g,m}\|_{HS}^2 
		&= \iint_{\R^\times\R^d} \left(\frac{f(x)}{|\eta|}- m\left(\frac{f(x)}{|\eta|}\right)\right)^2 \frac{dxd\eta}{(2\pi)^d} \\
		&= \int_{\R^d} f(x)^d\, dx \int_{\R^d} (|\eta|^{-1}-m(|\eta|^{-1}))^2\, \frac{d\eta}{(2\pi)^d} 
\end{align*}
Going to spherical coordinates shows 
\begin{align*}
	\int_{\R^d} (|\eta|^{-1}-m(|\eta|^{-1}))^2\, \frac{d\eta}{(2\pi)^d} 
		&= \frac{|S^{d-1}|}{(2\pi)^d} \int_0^\infty \left(r^{-1}-m(r^{-1})\right)^2 r^{d-1}\, dr \\
		&= \frac{ d|B_1^d|}{(2\pi)^d} \int_0^\infty (1-t^{-1}m(t))^2 t^{1-d}\, dt\, ,
\end{align*}
where $|S^{d-1}|$ is the surface area of the unit sphere in $\R^d$ and $|B_1^d|=|S^{d-1}|/d$ is the volume of the unit ball in $\R^d$. 

In particular, scaling $f$ by  $\kappa>0$, using 
$\kappa A_{f,g} =  A_{\kappa f,g}= B_{\kappa f,g,m}+ H_{\kappa f,g,m}$, the argument leading to \eqref{eq:first-upper-bound} now leads to
\begin{align}
	N(P^2-U)
	  & = n(A_{\kappa  f,g};\kappa)
	      \le (\kappa-\mu)^{-2}\sum_{j}\, \|H_{\kappa f,g,m}\|_{HS}^2 \label{eq:general-upper-bound-1} \\
	  &=  \frac{\kappa^d}{(\kappa-\mu)^{2}} \frac{ d|B_1^d|}{(2\pi)^d} \int_0^\infty (1-t^{-1}m(t))^2 t^{1-d}\, dt\, \int_{\R^d} U(x)^{d/2}\, dx \ ,  \label{eq:general-upper-bound-2}
\end{align}
as long as $\kappa> \mu\ge \|B_{\kappa f,g,m}\|$. 
It is important to note here that the last factor on the right hand side of the above bound has the correct dependence on the potential $U$. That is, the factor in front of it, which depends on the upper bound $\mu$ on the operator norm of $B_{f,g,m}$, \emph{has to be independent} of $f=\sqrt{U}$. 
Thanks to the second part of Theorem \ref{thm:great}, we can use 
$\mu = \|m_1(s)\|_{L^2(\R_+, \frac{ds}{s})} \|m_1(s)\|_{L^2(\R_+, \frac{ds}{s})}$ as an upper bound for $\|B_{f,g,m}\|$, which is independent of $f$ (and $g$), so the same bound holds for $\|B_{\kappa f,g,m}\|$ for any $\kappa>0$. Using this, we can now freely optimize in $\kappa>\mu$ in \eqref{eq:general-upper-bound-2}, to get 

\begin{align}\label{eq:super}
	N(P^2-U) \le C \frac{|B_1^d|}{(2\pi)^d} \int_{\R^d} U(x)^{d/2}\, dx
\end{align}
with the constant 
\begin{equation}\label{eq:super-constant}
  \begin{split}
	C&=C_{d,m}= \frac{d^{d+1}}{4(d-2)^{d-2}}\mu^{d-2} \int_0^\infty (1-t^{-1}m(t))^2 t^{1-d}\, dt  \\
		&= \frac{d^{d+1}}{4(d-2)^{d-2}}\left(\|m_1\|_{L^2(\R_+,\frac{ds}{s})}\|m_2\|_{L^2(\R_+,\frac{ds}{s})}\right)^{d-2}\int_0^\infty (1-t^{-1}m(t))^2 t^{1-d}\, dt
  \end{split}
\end{equation}
where $m(t)= m_1*m_2(t)=\int_0^\infty m_1(t/s)m_2(s)\, \frac{ds}{s}$.

This gives most of the main ideas of our proof of Theorem \ref{thm:CLR}. The last new idea, which is crucially important for the proof of \eqref{eq:B-bound}, is the connection between the bound on the norm of the operator $B_{f,g,m}$ and bounds for maximal Fourier multipliers on $L^2$. This is explained in Section \ref{sec:max-F-mult}.

\begin{remark}\label{rem:simple-choice}
 In order to get good bounds on $C$, we have to find good candidates for $m_1$ and $m_2$. A simple, but not optimal, choice is $m_1(s)= s\id_{\{0<s\le 1\}}$ and $m_2(s)= s^{-1}\id_{\{s>1\}}$, in which case  $\|m_1\|_{L^2(\R_+,\frac{ds}{s})}=\|m_2\|_{L^2(\R_+,\frac{ds}{s})}=1$ and  $m(t)= m_1*m_2(t)= \min(t,t^{-1})$, so 
 \begin{align*}
 	\int_0^\infty (1-t^{-1}m(t))^2 t^{1-d}\, dt
		= \int_1^\infty (1-t^{-2})^2 t^{1-d}\, dt
			= \frac{8}{(d-2)d(d+2)} \, .
 \end{align*}
 This gives 
 \begin{align*}
 	C_{0,d}= \frac{2\, d^d}{(d-2)^{d-1}(d+2)}
 \end{align*}
  as a possible constant in the CLR inequality and yields $C_{0,3}\le 10.8$, already an order of a magnitude smaller than Cwikel's bound. Moreover, combining this with `stripping-off-dimensions' ideas with the help of similar bounds for operator-valued potentials, one can get this bound also uniformly in  the dimension. 
To get the uniform bound claimed in Theorem \ref{thm:CLR} we have to choose better candidates for $m_1$ and $m_2$. We can achieve this in small dimensions and extend the bounds with the help of  bounds for operator-valued potentials, see Appendix \ref{sec:induction-in-dimension} and \ref{sec:numerics} and  Section \ref{sec:operator-valued}.
\end{remark}

Before we do this let us note a simple consequence of our approach for more general kinetic energies. 
 
\section{General kinetic energies}\label{sec:gen-kin-energy}
First we consider the case where  
$P^2$ is replaced by $P^{2\alpha}$ and give the 
\begin{proof}[Proof of Theorem \ref{thm:poly}]
  Replacing $g(\eta)=|\eta|^{-1}$ by $g(\eta)=|\eta|^{-\alpha}$ one simply reruns the argument from the previous section. Calculating, again by scaling, 
  \begin{align*}
	\|H_{f,g,m}\|_{HS}^2 
		&= \iint_{\R^\times\R^d} \left(\frac{f(x)}{|\eta|^\alpha}- m\left(\frac{f(x)}{|\eta|^\alpha}\right)\right)^2 \frac{dxd\eta}{(2\pi)^d} \\
		&= \int_{\R^d} f(x)^{d/\alpha}\, dx \int_{\R^d} (|\eta|^{-\alpha}-m(|\eta|^{-\alpha}))^2\, \frac{d\eta}{(2\pi)^d} 
\end{align*}
and 
  \begin{align*}
	\int_{\R^d} (|\eta|^{-\alpha}-m(|\eta|^{-\alpha}))^2\, \frac{d\eta}{(2\pi)^d} 
		&= \frac{|S^{d-1}|}{(2\pi)^d} \int_0^\infty \left(r^{-\alpha}-m(r^{-\alpha})\right)^2 r^{d-1}\, dr \\
		&= \frac{ d|B_1^d|}{\alpha (2\pi)^d} \int_0^\infty (1-t^{-1}m(t))^2 t^{1-\frac{d}{\alpha}}\, dt\, , 
  \end{align*}
  one sees that the argument leading to \eqref{eq:general-upper-bound-2} remains virtually unchanged, only  $d$ gets replaced by by $d/\alpha$. Thus  
  \begin{align*}
  	N(P^{2\alpha}+V)\le C \frac{ d|B_1^d|}{\alpha (2\pi)^d} \int_{\R^d} V_-(x)^{\frac{d}{2\alpha}}\, dx
  \end{align*}
  with constant 
  \begin{align*}
	C&=\frac{(\frac{d}{\alpha})^{\frac{d}{\alpha}+1}}{4(\frac{d}{\alpha}-2)^{\frac{d}{\alpha}-2}} 
		\left(\|m_1\|_{L^2(\R_+,\frac{ds}{s})}\|m_2\|_{L^2(\R_+,\frac{ds}{s})}\right)^{\frac{d}{\alpha}-2}\int_0^\infty (1-t^{-1}m(t))^2 t^{1-\frac{d}{\alpha}}\, dt
  \end{align*}
  For $m_1$ and $m_2$ we make the simple choice from Remark \ref{rem:simple-choice}. Then  $m(t)=m_1*m_2(t)= \min(t,t^{-1})$ and 
 $\mu=\|m_1\|_{L^2(\R_+,\frac{ds}{s})}\|m_2\|_{L^2(\R_+,\frac{ds}{s})}=1$. Hence,
 \begin{align*}
 	\int_0^\infty (1-t^{-1}m(t))^2 t^{1-\frac{d}{\alpha}}\, dt 
 		= \int_1^\infty (1-t^{-2})^2 t^{1-\frac{d}{\alpha}}\, dt 
 		= \frac{8}{(\frac{d}{\alpha}-2)\frac{d}{\alpha}(\frac{d}{\alpha}+2)}
 \end{align*}
 and collecting terms finishes the proof of Theorem \ref{thm:poly}. 
\end{proof}
\begin{remark}
	For the number of negative energy bound states of $P^{2\alpha}+U$ the so-far best bounds are due to Frank \cite{Fr,FrReview}. Using ideas from Rumin \cite{Ru1,Ru2}, he got the bound
	\begin{align*}
		N(P^{2\alpha}+V)
			\le \left(
					\frac{(\frac{d}{\alpha}(\frac{d}{\alpha+2})}{(\frac{d}{\alpha}-2)^2}
				\right)^{\frac{d}{2\alpha}-1} \frac{\frac{d}{\alpha}}{\frac{d}{\alpha}-2}
			\frac{|B_1^d|}{(2\pi)^d} \int_{\R^d} V_-(x)^{\frac{d}{2\alpha}}\, dx\, .
	\end{align*}
Even with the non-optimal choice of $m_1$ and $m_2$ above, a simple calculation shows that the bound from Theorem \ref{thm:poly}  is better as long as $2< \left( 1+2\alpha/d\right)^{d/(2\alpha)}$. Since $0<\delta\mapsto \left( 1+1/\delta\right)^{\delta} $ is strictly increasing, this is the case as soon as  $d>2\alpha$, that is, the whole range of allowed values of $\alpha$.
\end{remark}

For more general kinetic energies of the form $T(P)$ with $T$ a suitable non-negative measurable function obeying some mild growth condition at infinity, we have:

\begin{theorem}\label{thm:gen-kin-energy}
	The number of negative energy bound states of a Schr\"odinger--type operator $T(P)+V$, defined suitably with the help of quadratic form methods on $L^2$, obeys the  bound
	\begin{align}
		N(T(P)+V) \le \lambda^{-2} \int_{\R^d} G_{T}\big((\lambda+1)^2 V_-(x)\big)\, \d x
	\end{align}
	for any $\lambda>0$, with 
	\begin{align*}
		G_{T}(u) = \int_{\frac{u}{T}> 1} \left[\Big(\frac{u}{T(\eta)}\Big)^{1/2}- \Big(\frac{u}{T(\eta)}\Big)^{-1/2}\right]^2\, \frac{\d \eta}{(2\pi)^d}
	\end{align*}
	and $ V_-=\max(-V,0)$, the negative part of $V$. 
\end{theorem}
\begin{proof} 
In this case we use $g(\eta)= T(\eta)^{-1/2}$, $f(x)= V_-(x)$, and make again the  the choice 
 $m_1(s)=2s\id_{\{0<s\le 1\}}$ and $m_2(s)= s^{-1}\id_{s\ge 1}$. So  
 $\mu=\|m_1\|_{L^2(\R_+,\frac{ds}{s})}\|m_2\|_{L^2(\R_+,\frac{ds}{s})}=1$ and with $\lambda=\kappa-\mu=\kappa-1$, the same argument leading to  \eqref{eq:general-upper-bound-1} now gives 
 \begin{align*}
 	N(T(P)+V)\le N(T(P)-V_-) \le \lambda^{-2} \ \|H_{(\lambda+1)f,g,m}\|_{HS}^2\, .
 \end{align*}
 for any $\lambda>0$. 
 Using Theorem \ref{thm:great} to calculate the Hilbert--Schmidt norm shows
 \begin{align*}
 	\|H_{(\lambda+1)f,g,m}\|_{HS}^2 = \int_{\R^d} G_{T}\big((\lambda+1)^2 V_-(x)\big)\, \d x \ ,
 \end{align*}
 since $m(t)=m_1*m_2(t)=\min(t,t^{-1})$. 
\end{proof}
\begin{remarks} 
\itemthm The bound given in Theorem \ref{thm:gen-kin-energy} improves the bound from \cite{HHRW}, which was based on Cwikel's original method. 
  The result proven in \cite{HHRW} shows that under some rather mild general conditions on the kinetic energy symbol $T$ the operator 
  $T(P)+V$ has weakly coupled bound states for any non-trivial potential $V$, no matter how small, if $\eta\mapsto \frac{1}{T(\eta)}$ is 
  not integrable over the set $\{T<u\}$ for all small $u>0$, which is 
  equivalent to $G_T(u)=\infty$ for all small $u>0$. This shows that 
  the bound given by Theorem \ref{thm:CLR} is quite natural.  \\[0.2em]
\itemthm In all cases where one can explicitly calculate or find explicit upper bounds for $G_T$, Theorem \ref{thm:gen-kin-energy} gives an upper bound of the form 
  \begin{align*}
  	N(T(P)+V)\lesssim N^\text{cl}(T +V)
  \end{align*}
  with $ N^\text{cl}(T+V) $ the classical phase-space volume of the set where $H(\eta,x)=T(\eta)+V(x)<0$. See \cite{Solomyak,BL,BLS} and the discussion in Section 6 of \cite{HHRW}, where it is also shown that logarithmic corrections to the classical phase-space guess appear in critical cases. 
  \end{remarks}
\section{The connection with maximal Fourier multipliers}
\label{sec:max-F-mult}
In this section we finish the proof of Theorem \ref{thm:great}. Recall that given functions $f,g:\R^d\to[0,\infty)$ and a bounded function $m:\R_+\to\R_+ $, the operator $B_{f,g,m}$ is given by 
\begin{align}
	B_{f,g,m}\varphi(x) = (2\pi)^{-d/2} \int_{\R^d} e^{ix\eta} m(f(x)g(\eta)) \varphi(\eta)\, d\eta \ ,
\end{align}
at least for nice enough $\varphi$, say from the set of Schwartz functions. We would like to conclude that $B_{f,g,m}$ is a bounded operator on $L^2(\R^d)$, which might suggest to look for results when a pseudo-differential operator with symbol $a(x,\eta)= m(f(x)g(\eta))$ is bounded. But such results need enough differentiability of the symbol $a$, which we do not have. More importantly, we need an estimate independent of $f$, which one cannot get without looking more closely into the structure of the problem.  
To see how the product structure $f(x) g(\eta)$ helps in the operator bound, we rewrite $B_{f,g,m}$ as    
\begin{equation}
  \begin{split}
	B_{f,g,m}\varphi(x) 
		&= (2\pi)^{-d/2} \int_{\R^d} e^{ix\eta} m(tg(\eta)) \varphi(\eta)\, d\eta \, \Big\vert_{t=f(x)} \\
		&= \calF^{-1}\left[ m(tg(\cdot))\varphi(\cdot) \right](x)\Big\vert_{t=f(x)}
  \end{split}
\end{equation}
This suggest to look at the Fourier multiplier $B_{t,g,m}$ defined by 
\begin{align}\label{eq:F-mult}
	B_{t,g,m}\varphi \coloneqq \calF^{-1}\left[ m(tg(\cdot))\varphi(\cdot) \right]
\end{align}
and the associated maximal operator\footnote{To be more careful, one should take the supremum over a dense subset of $\R_+$, to ensure measurablity, but for $\varphi$ in the Schwarz class this makes no difference} 
\begin{align}\label{eq:max-F-mult} 
	B^*_{g,m}\varphi(x)\coloneqq \sup_{t>0}|B_{t,g,m}\varphi (x)|\,.
\end{align}
It is clear that, as operators on $L^2$, one has $\|B_{f,g,m}\|\le \|B^*_{g,m}\|$ for the corresponding operator norms. 
On the other hand, choosing $f(x)$ in such a way as to make $|B_{f,g,m}\varphi(x)|$ arbitrarily close to $B^*_{g,m}\varphi(x)$, shows the `reverse bound' 
$\sup_{f\ge 0}\|B_{f,g,m}\varphi\|_2 \ge \|B^*_{g,m}\varphi\|_2$ for a given fixed Schwartz function $\varphi$.
 Thus $\sup_{f\ge 0}\|B_{f,g,m}\|= \|B^*_{g,m}\|$, so having a bound for $B_{f,g,m}$ which is \emph{uniform} in the choice of the function  $f$ is \emph{ equivalent} to having a bound for the maximal Fourier multiplier $B^*_{g,m}$.   
 
 This is our starting point for the proof of the second half of Theorem \ref{thm:great}. 
\begin{theorem}\label{thm:max-op-bound}
	Let $g$ be a measurable non-negative function on $\R^d$ and assume that $m:\R_+\to\R $ is given by a convolution, 
	\begin{align*}
		m(t)= m_1*m_2(t)= \int_0^\infty m_1(t/s)m_2(s)\frac{ds}{s}
	\end{align*}
	with $m_1,m_2\in L^2(\R_+,\frac{ds}{s})$. 
	Then the maximal Fourier multiplier $B^*_{g,m}$, defined in \eqref{eq:max-F-mult}, is bounded  on $L^2(\R^d)$ with bound 
	\begin{align*}
		\|B^*_{g,m}\| \le \|m_1\|_{L^2(\R_+,\frac{ds}{s})} \|m_2\|_{L^2(\R_+,\frac{ds}{s})}
	\end{align*}
	for its operator norm. 
\end{theorem}

\begin{remark}
	There are several different but related proofs of boundedness of maximal Fourier multipliers available in the literature, see, e.g., \cite{Carbery,DappaTrebels,Rubio}. These works concentrate on getting $L^p$ bounds and do not care much about the involved constants. For us the $L^2$ boundedness is important, with good bounds on the operator norm.  
\end{remark}

\begin{proof}
  The proof is easy and uses nothing more than the Cauchy-Schwarz inequality: since $m$ is given by a convolution type integral, we have 
  \begin{align*}
  	B_{t,g,m}\varphi(x) 
  		& =	\int_0^\infty \calF^{-1}\left[ m_1(tg/s)\varphi \right](x) \, m_2(s)\, \frac{ds}{s}.   
  \end{align*}
Interchanging the integrals is certainly fine for nice enough $\varphi$ from a dense subset of $L^2$, say the Schwartz functions. Applying the triangle and then the Cauchy-Schwarz inequality for the $ds/s$ integration yields  
 \begin{align}
   |B_{t,g,m}\varphi(x)| 
  		& \le 	\int_0^\infty \left|\calF^{-1}\left[ m_1(tg/s)\varphi \right](x)\right| \, |m_2(s)|\, \frac{ds}{s}\nonumber\\ 
  		&\le \left( \int_0^\infty \left|\calF^{-1}\left[ m_1(tg/s)\varphi \right](x)\right|^2 \frac{ds}{s} \right)^{1/2} 
  			\|m_2\|_{L^2(\R_+,\frac{ds}{s})}\label{eq:punch}
 \end{align}
 The punchline is that the measure $ds/s$ is invariant under scaling, so we can scale $s$ by a fixed factor $t$ to see that 
 \begin{align*}
 	\int_0^\infty \left|\calF^{-1}\left[ m_1(tg/s)\varphi \right](x)\right|^2 \frac{ds}{s} 
 	= \int_0^\infty \left|\calF^{-1}\left[ m_1(g/s)\varphi \right](x)\right|^2 \frac{ds}{s}\, , 
 \end{align*}
 that is, the right hand side of \eqref{eq:punch} is independent of $t>0$. 
 So 
 \begin{align*}
 	B^*_{g,m}\varphi(x) = \sup_{t>0}|B_{t,g,m}\varphi(x)|
 		\le \left( \int_0^\infty \left|\calF^{-1}\left[ m_1(g/s)\varphi \right](x)\right|^2 \frac{ds}{s} \right)^{1/2} 
  			\|m_2\|_{L^2(\R_+,\frac{ds}{s})}.
 \end{align*}
 In particular, 
 \begin{align*}
 	\|B^*_{g,m}\varphi\|_2^2 
 		&\le  \|m_2\|_{L^2(\R_+,\frac{ds}{s})}^2 
 				\int_{\R^d} \int_0^\infty \left|\calF^{-1}\left[ m_1(g/s)\varphi \right](x)\right|^2 \, \frac{ds}{s}\, dx\, . 
 \end{align*}
Using Fubini--Tonelli to interchange the integrals and  Plancherel's theorem for the $L^2$ norm of the Fourier transform, one sees that
\begin{align*}
	\int_{\R^d} \int_0^\infty & \left|\calF^{-1}\left[ m_1(g/s)\varphi \right](x)\right|^2 \, \frac{ds}{s}\, dx\, 
		= \int_0^\infty \int_{\R^d}  |m_1(g(\eta)/s)|^2|\varphi(\eta)|^2 \, d\eta\,  \frac{ds}{s}\, \\
	&=  \int_0^\infty \int_{\R^d}  |m_1(s^{-1})|^2|\varphi(\eta)|^2 \, d\eta\,  \frac{ds}{s}\,
		= \| m_1\|_{L^2(\R_+,\frac{ds}{s})}^2 
			\|\varphi(\eta)\|_2^2   
\end{align*}
where we also used the same scaling argument\footnote{Strictly speaking one gets an inequality if $g(\eta)=0$ since  $m_1(0)=0$.} as before to scale out $g(\eta)$ and the invariance of $ds/s$ under inversion $s\mapsto s^{-1}$. Thus 
\begin{align*}
	\|B^*_{g,m}\varphi\|_2\le \| m_1\|_{L^2(\R_+,\frac{ds}{s})} \| m_2\|_{L^2(\R_+,\frac{ds}{s})} \|\varphi\|_2
\end{align*}
and we are done. 
\end{proof}
The next result, which also proves the second half of Theorem \ref{thm:great}, is an easy consequence of Theorem \ref{thm:max-op-bound}.  
\begin{corollary}
	Let $f,g$ be measurable non-negative functions on $\R^d$ and assume that $m:\R_+\to\R $ is given by a  convolution,
	\begin{align*}
		m(t)= m_1*m_2(t)= \int_0^\infty m_1(t/s)m_2(s)\frac{ds}{s}
	\end{align*}
	with $m_1,m_2\in L^2(\R_+,\frac{ds}{s})$. 
	Then the operator $B_{f,g,m}$, defined by \eqref{eq:kernel-B}, i.e.,  given by the kernel 
	\begin{align*}
		B_{f,g,m}(x,\eta) &= (2\pi)^{-d/2} e^{ix\cdot\eta}m(f(x)g(\eta))
	\end{align*}
	is bounded  on $L^2(\R^d)$ with bound 
	\begin{align*}
		\sup_{g\ge0}\big\|\sup_{f\ge0}|B_{f,g,m}\varphi|\big\|_2 \le \|m_1\|_{L^2(\R_+,\frac{ds}{s})} \|m_2\|_{L^2(\R_+,\frac{ds}{s})}\|\varphi\|_2 \, .
	\end{align*} 
\end{corollary}
\begin{proof} By definition of the maximal Fourier multiplier we have  $|B_{f,g,m}\varphi(x)|\le B^*_{g,m}\varphi(x)$ and thus also 
$\sup_{f\ge 0}|B_{f,g,m}\varphi(x)|\le B^*_{g,m}\varphi(x)$ for almost every $x\in\R^d$. 

Since the $L^2$--bound from Theorem \ref{thm:max-op-bound} is independent of $g\ge 0$, we can also take the supremum in $g\ge 0$ , after taking the $L^2$--norm. 
\end{proof}

\section{A lower bound for the variational problem}
\label{sec:lower-bound}
Recall that the variational problem, which comes up in a natural way in our bound on the number of bound states is 
\begin{align}
	P_\gamma= \inf \left\{ (\|m_1\|_{L^2(\R_+,\frac{ds}{s})}\|m_2\|_{L^2(\R_+,\frac{ds}{s})})^{\gamma-2} \int_0^\infty (1-t^{-1} m_1*m_2(t))^2 t^{1-\gamma}\,dt \right\},
\end{align}
where the convolution $m_1*m_2$ is on $\R_+$ with its scaling invariant measure $\frac{ds}{s}$, and the infimum is taken over all functions $m_1,m_2:\R_+\to\R $ .

\begin{theorem}
We have the lower bound 
  \begin{align*}
  	P_\gamma\ge \frac{2}{(\gamma-2)(\gamma-1)\gamma}
  \end{align*}
  for all $\gamma>2$. 
\end{theorem}
\begin{proof}
  The proof is straightforward, once one notices that, for example by Cauchy--Schwarz, one has $\|m\|_\infty \le \|m_1\|_{L^2(\R_+,\frac{ds}{s})}\|m_2\|_{L^2(\R_+,\frac{ds}{s})}$  for $m =m_1*m_2$. 
  Thus
  \begin{align*}
  	P_\gamma &\ge \inf_{m} \left\{ \|m\|_\infty^{\gamma-2} \int_0^\infty (1-t^{-1} m(t))^2 t^{1-\gamma}\,dt \right\} \\
  	&= \inf_{l>0}\inf_{\|m\|_\infty=l} \left\{ l^{\gamma-2} \int_0^\infty (1-t^{-1} m(t))^2 t^{1-\gamma}\,dt \right\} 
  \end{align*}
 To make $\int_0^\infty (1-t^{-1} m(t))^2 t^{1-\gamma}\,dt$ as small as possible under the constraint $\|m\|_\infty=l$, one simply chooses 
 $m(t)=\min(t,l)$. Then 
 \begin{align*}
 	\int_0^\infty (1-t^{-1} m(t))^2 t^{1-\gamma}\,dt
 	= \int_l^\infty (1-t^{-1}l)^2 t^{1-\gamma}\,dt 
 	= l^{2-\gamma} \frac{2}{(\gamma-2)(\gamma-1)\gamma}
 \end{align*}
 and one sees that
 \begin{align*}
 	\inf_{\|m\|_\infty=l} \left\{ l^{\gamma-2} \int_0^\infty (1-t^{-1} m(t))^2 t^{1-\gamma}\,dt \right\} 
 	= \frac{2}{(\gamma-2)(\gamma-1)\gamma}. 
 \end{align*}
\end{proof}

\section{Extension to operator--valued potentials}\label{sec:operator-valued}
In this section we extend our method to operator--valued potentials and give the proof of Theorem \ref{thm:poly-operator-valued}, i.e. we prove that the number of negative bound states of $P^{2\alpha}\otimes \idG + V$ is bounded by 
 \begin{align*}
 	N(P^{2\alpha}\otimes\idG+V)
			\le C_{d/\alpha}\, 
			\frac{|B_1^d|}{(2\pi)^d} \int_{\R^d} \tr_\calG[ V_-(x)^{\frac{d}{2\alpha}}]\, dx \ ,
 \end{align*}
 where $V:\R^d\to \calB(\calG)$ is an operator valued potential with positive part $V_+\in L^1_{\text{loc}}(\R^d,\calB(\calG))$ and negative part  $V_-\in L^{d/(2\alpha)}(\R^d, \calS_{d/(2\alpha)}(\calG))$.
 
 Let $U(x)= V(x)_-$, the negative part of $V(x)$ defined by spectral calculus. 
The Birman--Schwinger operator corresponding to $|P|^{2\alpha}\otimes\idG -U$ is given by 
\begin{align*}
	K= \sqrt{U} (|P|^{-2\alpha}\otimes\idG) \sqrt{U}
\end{align*}
and we again have
\begin{align*}
	N(|P|^{2\alpha}\otimes\idG +V) \le N(|P|^{2\alpha}\otimes\idG -U)
		= n(K; 1).
\end{align*}
Now we factor $K$ as $K= \wti{A}_{f,g}^* \wti{A}_{f,g}$ where $\wti{A}_{f,g}$ has kernel 
\begin{align*}
	\wti{A}_{f,g}\varphi (\eta) 
		= (2\pi)^{-d/2} \int_{\R^d} e^{-i\eta \cdot x} g(\eta) f(x)\varphi(x)\, dx \ ,
\end{align*}
$g(\eta)= |\eta|^{-\alpha}$ is real--valued (even positive), and $f(x)=\sqrt{U(x)}$ takes values in the self-adjoint positive operators on $\calG$. We split this as 
\begin{align*}
	\wti{A}_{f,g}= \wti{B}_{f,g,m} + \wti{H}_{f,g,m}
\end{align*}
with a function $m:[0,\infty)\to \R $, so that
\begin{align}\label{eq:kernel-B-operator-valued}
	\wti{B}_{f,g,m}\varphi(\eta) &= (2\pi)^{-d/2} \int_{\R^d} e^{-i\eta \cdot x} m(g(\eta) f(x))\varphi(x)\, dx 
		= \calF\left[ m(tf )\varphi \right](\eta)\Big|_{t=g(\eta)}
\end{align}
and 
\begin{align}\label{eq:kernel-H-operator-valued}
	\wti{H}_{f,g,m}\varphi(\eta) &= (2\pi)^{-d/2} \int_{\R^d} e^{-i\eta \cdot x} \left[g(\eta)f(x) -m(g(\eta) f(x))\right]\varphi(x)\, dx \ , 
\end{align}
where $\varphi$ is a function from a nice dense subset of $L^2(\R^d,\calG)$, so that the integrals converge and $m(t f(x))$ is an operator on $\calG$ defined via functional calculus. 
\begin{remark}
  With a slight abuse of notation, we write $\calF$ in the definition of $\wti{B}_{f,g,m}$, which strictly speaking denotes the Fourier transform on $L^2(\R^d)$, instead of $\calF\otimes\idG$, the Fourier transform on $L^2(\R^d,\calG) = L^2(\R^d)\otimes\calG$.  In addition, in the definition of $\wti{B}_{f,g,m}$ and 
	$\wti{H} _{f,g,m}$ above we swapped the role of $f$ and $g$ compared to the discussion in Section \ref{sec:max-F-mult}. This is convenient, since by assumption $g(\eta)$ is a  multiplication operator  on $\calG$, and this makes a maximal Fourier multiplier estimate, now with $g$ instead of $f$, easier. 
	The general case can be reduced to this setting, see Section \ref{sec:trace-ideals} below.
\end{remark}

The following theorem is the replacement of Theorem \ref{thm:great} in the operator-valued setting.

\begin{theorem}\label{thm:great-operator-valued}  
	 $\wti{H}_{f,g,m}$ is  Hilbert--Schmidt operator on $\calH=L^2(\R^d,\calG)$ with Hilbert--Schmidt norm given by  
	\begin{align}\label{eq:operator-valued-HS-norm}
		\|\wti{H}_{f,g,m}\|_{\calS_2(\calH)}^2
		= \int_{\R^d} \tr_\calG\left[G_{g,m}(f(x))\right]\, dx \ ,
	\end{align}
	where $G_{g,m}$ is again given by 
	\begin{align}
		G_{g,m}(u)= \int_{\R^d} |ug(\eta)-m(ug(\eta))|^2 \frac{\d \eta}{(2\pi)^d} .
	\end{align}
    If, moreover, $m=m_1*m_2$   
    then for all measurable non-negative functions $g$ and non-negative operator-valued functions $f$ the operator $\wti{B}_{f,g,m}$ is bounded on $\calH$ with 
    \begin{align}
    	\|\wti{B}_{f,g,m}\varphi\|_{\calH} 
    		\le  \|m_1\|_{L^2(\R_+, \frac{ds}{s})} 
    			\|m_2\|_{L^2(\R_+, \frac{ds}{s})} 
    				\|\varphi\|_{\calH}
    \end{align}
    for all $\varphi\in \calH$. 
\end{theorem}
\begin{proof}
  To prove \eqref{eq:operator-valued-HS-norm}, we note that the Hilbert--Schmidt operators on $\calH= L^2(\R^d,\calG)$ are isomorphic to operators with kernels in $L^2(\R^d\times\R^d, \calS_2(\calG))$ and 
  \begin{align*}
  	\|\wti{H}\|_{\calS_2(\calH)}^2 &= \tr_{\calH}\big[ \wti{H}^* \wti{H} \big] 
  		= \iint_{\R^d\times\R^d} \|\wti{H}(\eta,x)\|_{\calS_2(\calG)}^2 \, d x d\eta \ ,
  \end{align*}
see Lemma \ref{lem:hilbert-schmidt}.
	
	Using the explicit form of the `kernel' of $\wti{H}_{f,g,m}$ given in \eqref{eq:kernel-H-operator-valued} this shows  
 \begin{align*}
 	\|\wti{H}\|_{\calS_2(\calH)}^2 
 		&= (2\pi)^{-d}\int_{\R^d}\int_{\R^d} \tr_\calG\big[ |g(\eta)f(x)- m(g(\eta f(x)))|^2 \big]\, d\eta d x \\
 		&= \int_{\R^d} \tr_\calG\big[ G_{g,m}(f(x)) \big]\, d x 
 \end{align*}
 by the definition of $G_{g,m}$ and the spectral theorem. 
 
 Concerning the boundedness of $\wti{B}_{f,g,m}$ we recall \eqref{eq:kernel-B-operator-valued} and, if $m=m_1*m_2$, 
 \begin{align*}
 	\wti{B}_{f,t,m}\varphi(\eta) 
 		\coloneqq \calF\left[ m(tf)\varphi \right](\eta) 
 		= \int_0^\infty \calF\left[ m_1(f /s)\varphi \right](\eta) \, m_2(ts)\, \frac{ds}{s}.
 \end{align*}
 Thus, 
 \begin{align*}
 	\big\| \wti{B}_{f,t,m}\varphi(\eta) \big\|_\calG 
 		&\le  \int_0^\infty \big\| \calF\left[ m_1(f /s)\varphi \right](\eta) \big\|_\calG\, |m_2(ts)|\, \frac{ds}{s} \\
 		&\le \left( \int_0^\infty \big\| \calF\left[ m_1(f /s)\varphi \right](\eta) \big\|_\calG^2 \,\frac{ds}{s}\right)^{1/2} 
 			\left( \int_0^\infty |m_2(ts)|^2\, \frac{ds}{s} \right)^{1/2} \\
 		&= \left( \int_0^\infty \big\| \calF\left[ m_1(f /s)\varphi \right](\eta) \big\|_\calG^2 \,\frac{ds}{s}\right)^{1/2} \|m_2\|_{L^2(\R_+,\frac{ds}{s})}
 \end{align*}
 due to the scaling invariance of $ds/s$. We therefore have a maximal operator bound 
 \begin{align*}
 	\wti{B}_{f,m}^* \varphi(\eta)
 		\coloneqq \sup_{t>0} 
 					\big\| \wti{B}_{f,t,m}\varphi(\eta) \big\|_\calG 
 					\le \left( \int_0^\infty \big\| \calF\left[ m_1(f /s)\varphi \right](\eta) \big\|_\calG^2 \,\frac{ds}{s}\right)^{1/2} \|m_2\|_{L^2(\R_+,\frac{ds}{s})} \ .
 \end{align*}
In particular, 
 \begin{align*}
 	\|\wti{B}_{f,m}^* \varphi\|_{L^2(\R^d)}^2 
 	\le \|m_2\|_{L^2(\R_+.\frac{ds}{s})}^2 \int_{\R^d} \int_0^\infty \big\| \calF\left[ m_1(f /s)\varphi \right](\eta) \big\|_\calG^2 \,\frac{ds}{s}\, d\eta \ ,
 \end{align*}
and  
 \begin{align*}
 	\int_{\R^d} \int_0^\infty &\big\| \calF\left[ m_1(f /s)\varphi \right](\eta) \big\|_\calG^2 \,\frac{ds}{s}\, d\eta 
 	 = \int_0^\infty \int_{\R^d} \left\la \calF[m_1(f/s)\varphi](\eta), \calF[m_1(f/s)\varphi](\eta) \right\ra_\calG \, d\eta\, \frac{d s}{s} \\
 	&= \int_0^\infty \int_{\R^d} \left\la m_1(f(x)/s)\varphi(x), m_1(f(x)/s)\varphi(x) \right\ra_\calG \, d x\, \frac{d s}{s} \\
 	&= \int_{\R^d} \left\la \varphi(x), \int_0^\infty m_1(f(x)/s)^2\, \frac{ds}{s} \varphi(x) \right\ra_\calG\, dx \\
 	 &= \|m_1\|_{L^2(\R_+, \frac{ds}{s})}^2 \int_{\R^d} \|\varphi(x)\|_\calG^2\, dx 
 	 	= \|m_1\|_{L^2(\R_+, \frac{ds}{s})}^2 
 	 		\|\varphi\|_\calH^2
 \end{align*}
 where we again used that, by scaling $\int_0^\infty m_1(r/s)^2\, \frac{ds}{s}= \|m_1\|_{L^2(\R_+, \frac{ds}{s})}^2$ for all $r>0$, so by functional calculus 
 \begin{align*}
	\int_0^\infty m_1(f(x)/s)^2\, \frac{ds}{s}
	= \|m_1\|_{L^2(\R_+, \frac{ds}{s})}^2\id_\calG\, .
 \end{align*}
 Altogether, we get the operator-valued version of our previous maximal Fourier multiplier bound in the form  
 \begin{align*}
 	\|\wti{B}_{f,m}^* \varphi\|_{L^2(\R^d)}^2
 	\le \|m_1\|_{L^2(\R_+, \frac{ds}{s})} 
 		\|m_2\|_{L^2(\R_+, \frac{ds}{s})} 
 	 		\|\varphi\|_\calH \ ,
 \end{align*}
 and it is easy to see that 
 \begin{align*}
 	\|\wti{B}_{f,g,m}\varphi\|_{\calH} \le \|\wti{B}_{f,m}^* \varphi\|_{L^2(\R^d)} \ ,
 \end{align*}
 which completes the proof of Theorem \ref{thm:great-operator-valued}.
\end{proof}

The proof of Theorem \ref{thm:poly-operator-valued} is now easy: one simply does the same steps as in the scalar case with \eqref{eq:general-upper-bound-1} replaced by
\begin{align*}
		N(P^{2\alpha}\otimes\idG-U)
	  & = n(\wti{A}_{\kappa f, g};\kappa)
	      \le (\kappa-\mu)^{-2}\sum_{j}\, \|\wti{H}_{\kappa f,g,m}\|_{\calS_2(\calH)}^2 \ ,
\end{align*}
where now $\mu \ge \|\wti{B}_{\kappa f,g,m}\varphi\|_{\calH}$. 
As before, Theorem \ref{thm:great-operator-valued} gives a $\kappa$-independent bound for $\|\wti{B}_{\kappa f,g,m}\varphi\|_{\calH}$,  in particular, we can take any  $\mu\ge \|m_1\|_{L^2(\R_+, \frac{ds}{s})} \|m_2\|_{L^2(\R_+, \frac{ds}{s})} $. It also allows us to calculate the Hilbert-Schmidt norm.  
For $g(\eta)=|\eta|^{-\alpha}$ we get 
\begin{align*}
	G_{g,m}(u) = u^{d/\alpha} \int_{\R^d} (|\eta|^{-\alpha}-m(|\eta|^{-\alpha}))^2\, \frac{d\eta}{(2\pi)^d} \ ,
\end{align*}
so 
\begin{align*}
	\|\wti{H}_{\kappa f,g,m}\|_{\calS_2(\calH)}^2 
	= \kappa^{d/\alpha} \int_{\R^d} (|\eta|^{-\alpha}-m(|\eta|^{-\alpha}))^2\, \frac{d\eta}{(2\pi)^d}
		\int_{\R^d} \tr_\calG\left[ f(x)^{d/\alpha} \right]\, dx \ .
\end{align*}
Using this in the above bound for $N(P^{2\alpha}\otimes\idG-U)$ and 
minimizing over $\kappa$, as in the scalar case, finishes the proof of Theorem \ref{thm:poly-operator-valued}.

\section{Trace ideal estimates}\label{sec:trace-ideals}

In this section we show how the ideas developed so far give a simple proof of a fully operator-valued version of Cwikel's theorem. Such an inequality was first proved in \cite{Fr}.

In this setting let $(X,dx)$ and $(Y,dy)$ be sigma-finite measure spaces and $\calH, \calG$ (separable) Hilbert spaces. We denote by $L^p(X,\calS_p(\calH))$ the set of measurable functions $f: X\to \calS_p(\calH)$, where $\calS_p(\calH)$ is the space of $p$-summable compact operators, i.e. the von Neumann--Schatten class, on $\calH$, such that  
\begin{align*} 
	\|f\|_{L^p(X,\calS_p(\calH))}^p \coloneqq \int_{X} \|f(x)\|_{\calS_p(\calH)}^p\, dx <\infty \ .
\end{align*}
Similarly, we denote by $L^p_\text{w}(Y, \calB(\calG))$ the set of of all measurable functions $g:Y\to \calB(\calG)$, with values in the bounded operators on $\calG$, such that 
\begin{align*}
	\|g\|_{L^p_\text{w}(Y, \calB(\calG))}^p 
		\coloneqq \sup_{t>0} t^p\left|\left\{ y\in Y:\, \|g(y)\|_{\calB(\calG)>t}\right\} \right| <\infty
\end{align*}
A map $A: L^2(X,\calH)\to L^2(Y,\calG)$ is in the weak trace--ideal $\calS_{p,\text{w}}= \calS_{p,\text{w}}(L^2(X,\calH), L^2(Y,\calG))$ if 
\begin{align}
	\left\| f\Phi^* g \right\|_{p,\text{w}} 
		\coloneqq \Big(\sup_{n\in\N} n^p s_n(A)\Big)^{1/p}
\end{align}
where $s_n(A)$ are teh singular--values of $A$, i.e. the eigenvalues of $A^* A:L^2(X,\calH)\to L^2(X,\calH)$.

\begin{theorem}[Fully operator valued version of Cwikel's theorem]\label{thm:Cwikel}
  Let $\Phi: L^2(X,\calH)\to L^2(Y,\calG)$ be a unitary operator, 
  which is also bounded from  $L^1(X,\calH)$ into $L^\infty(Y,\calG)$. 
  If $p>2$ and $f\in L^p(X,\calS_p(\calH))$ and $g\in L^p_\text{w}(Y, \calB(\calG))$, then $f\Phi^* g$ is in the weak trace ideal 
  $\calS_p(L^2(X,\calH),L^2(Y,\calG))$ and 
  \begin{align*}
  \begin{split}
  	\left\| f\Phi^* g \right\|_{p,\text{w}}^p 
  		&\coloneqq \sup_{n\in\N} n^p s_n(f\Phi^* g) \\
  		&\le \frac{2(p-2)}{p+2} \, \left(\frac{p}{p-2} \right)^{p} \|\Phi\|_{L^1\to L^\infty}^2 \|f\|_{L^p(X,\calS_p(\calH))}^p \|g\|_{L^p_\text{w}(Y, \calB(\calG))}^p \,. 
  \end{split}
  \end{align*} 
\end{theorem}

\begin{remark}
	Theorem \ref{thm:Cwikel} improves the result of Frank in \cite{Fr},  
	\begin{align*}
		\left\| f\Phi^* g \right\|_{p,\text{w}}^p 
			\le \frac{p}{2} \left(\frac{p}{p-2} \right)^{p-1} 
				\|\Phi\|_{L^1\to L^\infty}^2 \|f\|_{L^p(X,\calS_p(\calH))}^p \|g\|_{L^p_\text{w}(Y, \calB(\calG))}^p \, ,
	\end{align*}
 by a factor of $(p+2)/2$. In addition, his bound in the scalar case, when $\Phi$ is the usual Fourier-transform,  is worse than Theorem \ref{thm:Cwikel} by a factor of $\frac{1}{2}(1+p/2)^{p/2}>1$ in the allowed range $p>2$.  
\end{remark}

\begin{proof}
 First we note that one can reduce the result to the case when $g$ is pointwise a positive multiple of the identity operator on $\calG$. As operators on $\calG$ one has  $g(y) g(y)^* \le \|g(y)\|_{\calB(\calG)}^2\id_{\calG}$. Thus 
  with $A_1= f\Phi^* g$ we have 
  \begin{align*}
  	A_1 A_1^* = f\Phi^* g g^* \Phi f^* \le  f\Phi^* (\|g\|_{\calB(\calG)}\id_{\calG})^2 \Phi f^* = A_2 A_2^*
  \end{align*}
 with $A_2= f\Phi^* \|g\|_{\calB(\calG)}\id_{\calG}= f\Phi^* \|g\|_{\calB(\calG)} $ where, for simplicity, we   wrote   
 $\|g\|_{\calB(\calG)}$ for $\|g\|_{\calB(\calG)}\id_\calG$. 
 Since   the singular values of $A_1$ are the square roots of the eigenvalues of $A_1^* A_1$, which has the same non-zero-eigenvalues as $ A_1 A_1^*$ we see that the nonzero singular values of $A_1$ obey the bound 
 $s_n(A_1) \le s_n(A_2)$. 
 
 Similarly, $|f(x)|\coloneqq \sqrt{f(x)^*f(x)}$  is a non negative operator on 
 $\calH$ and 
 \begin{align*}
 	A_2^* A_2 = \|g\|_{\calB(\calG)} \Phi^* f^*f\Phi^* \|g\|_{\calB(\calG)} 
 	= \|g\|_{\calB(\calG)} \Phi^* |f|^2\Phi^* \|g\|_{\calB(\calG)} 
 	= A_3^* A_3
 \end{align*}
 with $ A_3= |f|\Phi^* \|g\|_{\calB(\calG)}$. So the singular 
 values of $A_2$ are the same as the singular values of $A_3$ and 
 without loss of generality, we can assume that $g$ is a non-negative 
 function and $f$ takes values in the non-negative operators on $\calH$. 
 By scaling, we can also assume that $\|f\|_{L^p(X,\calS_p(\calH))}= \|g\|_{L^p_\text{w}(Y)}^p=1$. 
 
 Since $\Phi: L^1(X,\calH)\to L^\infty(Y,\calG)$ is bounded, Lemma \ref{lem:kernel-op-valued} shows that it has a kernel $\Phi(\cdot,\cdot)$ such that for all $f\in L^2(X,\calH)$, 
 \begin{align*}
 \Phi f(y) = \int_X \Phi(y,x) f(x) \, dx  
 \end{align*}
 for almost all $y\in Y$. Moreover, $\sup_{(y,x)\in Y\times X}\|\Phi(y,x)\|_{\calB(\calH,\calG)}= \|\Phi\|_{L^1\to L^\infty}$. 
 Having reduced the estimate to scalar non-negative functions $g$ and non-negative operator-valued functions $f$ we  can rewrite  $A_{f,g}= g\Phi f$ as
 \begin{align}
 	A_{f,g}\varphi(y) = \int_{X} g(y) \Phi(y,x) f(x) \varphi(x)\, dx 
 		= \int_{X}  \Phi(y,x) g(y) f(x) \varphi(x)\, dx
 \end{align}
 using that $g(y)$ is now a non-negative scalar. Thus, we can take again an arbitrary function $m:\R_+\to \R $ with $m(0)=0$ and split 
  \begin{align}
  	 B_{f,g,m}\varphi(y) 
  	 	& \coloneqq 
  	 		\int_{X}  \Phi(y,x) m\big( g(y) f(x) \big) \varphi(x)\, dx 
  	 			\label{eq:def-B-fully-operator-valued} \\
  	 H_{f,g,m}\varphi(y)	
  	 	& \coloneqq 
  	 		\int_{X}  \Phi(y,x) \big[g(y)f(x)-m\big( g(y) f(x) \big)\big] \varphi(x)\, dx \label{eq:def-H-fully-operator-valued}
  \end{align}
  The above expression are well-defined by the spectral theorem, since $g$ is a non-negative function and $f$ takes values in the non-negative operators on $\calH$, so $m(g(y)f(x))$ is a bounded operator on $\calH$ for almost all $y$ and $x$, when $m$ is bounded. Thus the integrals in \eqref{eq:def-B-fully-operator-valued} and \eqref{eq:def-B-fully-operator-valued}  converge for all $\varphi$ from a dense subset of $L^2(X,\calH)$, for example the piecewise constant functions.

  Scaling in $f$ by $\kappa>0$, we get from Ky Fan's inequality
  \begin{equation}
    \begin{split}  	
    	s_n(g\Phi f) 
    		&= \kappa^{-1}s_n(A_{\kappa,f,g})\le \kappa^{-1}\left[ \|B_{\kappa f,g,m}\| + s_n(H_{\kappa f,g,m}) \right] \\ 
    		& \le \kappa^{-1} \left[ \mu + n^{-1/2} \|H_{\kappa, f,g,m}\|_{HS} \right]
    \end{split} 
  \end{equation}
      where we take $\mu = \|m_1\|_{L^2(\R_+,\frac{ds}{s})} \|m_2\|_{L^2(\R_+,\frac{ds}{s})}$, the upper bound on the norm of $B_{\kappa f,g,m}$ from Lemma \ref{lem:great-fully-operator-valued} below and we used  
      $s_n(H)\le n^{-1}\sum_{j=1}^n s_j(H)^2\le n^{-1}\|H\|_{HS}^2$, for any Hilbert-Schmidt operator, due to the monotonicity of its singular values.  
  Thus using the bound \eqref{eq:bound-H-fully-operator-valued} one gets  
  \begin{align*}
  	    	s_n(g\Phi f)  
    			 \le \kappa^{-1} \left[ \mu + n^{-1/2} D^{1/2} \kappa^{p/2} \right]
  \end{align*}
 with $D= p \, \|\Phi\|_{L^1\to L^\infty}^2 \int_0^\infty (1-t^{-1}m(t))^2 t^{1-p}\, dt$, and minimizing this over $\kappa>0$ we have 
 \begin{align*}
 	s_n(g\Phi f) 
 		\le \frac{p}{p-2} \left( \frac{p-2}{2} \right)^{2/p} 
 			(\mu^{p-2}D)^{1/p}\, n^{-1/p}
 \end{align*}
 for the singular values for all $n\in\N$. 
 
 Making again the simplest choice $m_1(s)= 2s\id_{\{0\le s\le 1\}}$ and $m_2(s)= s^{-1}\id_{\{s\ge 1\}}$ one checks that $m=m_1*m_2= \min(t,t^{-1})$ is allowed since $m'\le 1$. Calculating the numbers finishes the proof of Theorem \ref{thm:Cwikel}.
\end{proof}

\begin{lemma}\label{lem:great-fully-operator-valued}
  Let $p>2$, $\calH$ and $\calG$ auxiliary Hilbert--spaces, $(X,dx)$ and 
  $(Y,dy)$ $\sigma$--finite measure spaces, $0\le g\in L^p_{\text{w}}(Y)$, 
  $0\le f\in L^p(X\calS_p(\calH))$, $\Phi: L^2(X,\calH)\to L^2(Y,\calG)$ unitary 
  and also bounded from  $L^1(X,\calH)\to L^\infty(Y,\calG)$. Then 
  for all continuous and piecewise differentiable bounded functions $m:\R_+\to \R $ with  $m(0)=0$ and $m'\le 1$ the 
  operator $\wti{H}_{f,g,m}$ defined in \eqref{eq:def-H-fully-operator-valued} is 
  a Hilbert--Schmidt operator and 
  \begin{equation}
  \begin{split}\label{eq:bound-H-fully-operator-valued}
  	\|\wti{H}_{f,g,m}&\|_{\calS_2(L^2(X,\calH)\to L^2(Y,\calG))}^2 
  		= \tr_{L^2(X,\calH)}\left[ \wti{H}_{f,g,m}^*\wti{H}_{f,g,m} \right] \\
  	&\le p \, \|\Phi\|_{L^1\to L^\infty}^2 \int_0^\infty (1-t^{-1}m(t))^2 t^{1-p}\, dt \, 
  				\|g\|_{L^p_\text{w}(Y)}^p \| f\|_{L^p(X,S_p(\calH))}^p\,.
  \end{split}  	
  \end{equation} 
  
  Moreover, if $m=m_1*m_2$, then the operator 
  $\wti{B}_{f,g,m}$ defined in \eqref{eq:def-B-fully-operator-valued} is 
  bounded from $L^2(X,\calH)$ to $ L^2(Y,\calG)$ and 
  \begin{align}\label{eq:bound-B-fully-operator-valued}
  	\|\wti{B}_{f,g,m}\|_{L^2\to L^2} \le \|m_1\|_{L^2(\R_+,\frac{ds}{s})} \|m_2\|_{L^2(\R_+,\frac{ds}{s})}\, .
  \end{align}
\end{lemma}
\begin{remark}
  As the proof of Lemma \ref{lem:great-fully-operator-valued} shows one even has a bound on  $\wti{B}_{f,g,m}$ in the form   
  \begin{align*}
  	\sup_{f\ge 0} \big\| \sup_{g\ge 0}\|\wti{B}_{f,g,m}\varphi\|_\calG \big\|_{L^2(Y)} 
  		\le \|m_1\|_{L^2(\R_+,\frac{ds}{s})} \|m_2\|_{L^2(\R_+,\frac{ds}{s})}^2
 		 		\|\varphi\|_{L^2(X,\calH)}
  \end{align*}
 	where the first supremum is taken over all functions $g:Y\to [0,\infty)$ and the second supremum is taken over all non-negative operator-valued functions $f:X\to \calB(\calH)$. 
\end{remark}

\begin{proof}
  For notational simplicity we set  
  \begin{align*}
  	C= \|\Phi\|_{L^1(X,\calH)\to L^\infty(Y,\calG)} = \mathop{\mathrm{esssup}}\limits_{(y,x)\in Y\times X} \|\Phi(x,y)\|_{\calB(\calH,\calG)}\, .
  \end{align*}
 and note  
 \begin{align*}
 	\|\wti{H}_{f,g,m}\|_{\calS_2(L^2(X,\calH)\to L^2(Y,\calG))}^2 
 		&= \iint_{Y\times X} \tr_{\calH}\left[ \wti{H}_{f,g,m}(y,x)^*\wti{H}_{f,g,m}(y,x) \right]\, dy dx \, .
 \end{align*}
 Because $g$ is real-valued, even positive, and $f$ takes values in the non-negative, hence self-adjoint, operators 
 \begin{align*}
 	\wti{H}_{f,g,m}&(y,x)^*\wti{H}_{f,g,m}(y,x) =\\
 		&= \big[g(y)f(x)-m\big( g(y) f(x) \big)\big]\Phi(y,x)^*\Phi(y,x)\big[g(y)f(x)-m\big( g(y) f(x) \big)\big] \\
 		&\le C^2 \big[g(y)f(x)-m\big( g(y) f(x) \big)\big]^2,  
 \end{align*}
 so, setting $
 	G(u)\coloneqq \int_{Y}\left[u g(y) -m(ug(y)) \right]^2\, dy$, 
 we have 
 \begin{align*}
 	\|\wti{H}_{f,g,m}\|_{\calS_2(L^2(X,\calH)\to L^2(Y,\calG))}^2 
 		&\le  C^2 \int_{X} \tr_\calH G(f(x)) \, dx \, .
 \end{align*}
 With  $k(t)= (t-m(t))^2 $, the layer-cake principle shows 
 \begin{align*}
 	G(u) = \int_0^\infty k'(ut)u |\{y\in Y: g(y)>t\}|\, dt\, .
 \end{align*}
 By definition $|\{y\in Y: g(y)>t\}|\le t^{-p}\|g\|_{L^p_\text{w}(Y)}^p$ 
 for all $t>0$. If $m(0)=0$ and $m'(t)\le 1$, then $k'\ge 0$. Thus 
  \begin{align*}
 	G(u) &\le  u^p\,\|g\|_{L^p_\text{w}(Y)}^p 
 				\int_0^\infty k'(t) t^{-p}\, dt\, . 
 \end{align*}
 For any $0<\veps<L$ we have 
 \begin{align*}
 	\int_\veps^L k'(ut)u t^{-p}\, dt
 		= p[k(t)t^{-1-p}]_\veps^L +p \int_\veps^L k(t) t^{-1-p}\, dt 
 \end{align*}
 and if $0<t \mapsto k(t) t^{-1-p}\in L^1(\R_+,dt)$, 
 then there exist  sequences $\veps_n\to 0$ and $L_n\to\infty$ 
 such that  $\lim_{n\to\infty}k(\veps_n) \veps_n^{-1-p}=0 = \lim_{n\to\infty}k(L_n) L_n^{-1-p}$, 
 that is, the boundary term $ [k(t)t^{-1-p}]_{\veps_n}^{L_n}$ vanishes in the limit $n\to\infty$. Hence integration by parts is 
 justified as soon as the right hand side of \eqref{eq:bound-H-fully-operator-valued} is finite and 
 \begin{align*}
 	\tr_\calH G(f(x)) \le 
 		p \int_0^\infty k(t) t^{-1-p}\, dt\, \|g\|_{L^p_\text{w}(Y)}^p \tr_\calH (f(x)^p)
 \end{align*}
Integrating this over $X$ finishes the proof of \eqref{eq:bound-H-fully-operator-valued}. 
 \smallskip
 
 To prove \eqref{eq:bound-B-fully-operator-valued} we introduce 
 \begin{align}\label{eq:def-maximal-operator-fully-operator-valued}
 	\wti{B}_{f,t,m}\varphi(y)
 		\coloneqq \int_{X}  \Phi(y,x) m\big( t f(x) \big) \varphi(x)\, dx
 		= \Phi[m(tf)\varphi](y)
 \end{align}
 for $t\ge 0$ (using $m(0)=0$). If $m=m_1*m_2$, convolution on $\R_+$, then a by now familiar calculation yields
 \begin{align*}
  	\wti{B}_{f,t,m}\varphi(y)
 		= 
 			\int_0^\infty \Phi[m_1(sf)\varphi](y) \, m_2(ts)\,\frac{ds}{s}
 \end{align*}
 and therefore the Cauchy--Schwarz inequality gives 
 \begin{align*}
  	\|\wti{B}_{f,t,m}\varphi(y)\|_\calG
 		&\le  
 			\int_0^\infty \|\Phi[m_1(sf)\varphi](y)\|_\calG \, |m_2(ts)|\,\frac{ds}{s} \\
 		&\le \left(  \int_0^\infty \|\Phi[m_1(sf)\varphi](y)\|_\calG ^2 \,\frac{ds}{s} \right)^{1/2}
 			\left(  \int_0^\infty  |m_2(ts)|^2\,\frac{ds}{s} \right)^{1/2} \, .
 \end{align*}
 By scaling, the right hand side above does not depend on $t>0$ anymore and, since $\wti{B}_{f,0,m}\varphi(y)=0$, we get the bound  
 \begin{align*}
  	\wti{B}_{f,m}^*\varphi(y) 
  		= \sup_{t> 0}\|\wti{B}_{f,t,m}\varphi(y)\|_\calG 
 		\le \|m_2\|_{L^2(\R_+,\frac{ds}{s})} \left(  \int_0^\infty \|\Phi[m_1(sf)\varphi](y)\|_\calG ^2 \,\frac{ds}{s} \right)^{1/2} .
 \end{align*}
 for the associated maximal operator 
 $\wti{B}_{f,m}^*\varphi(y)\coloneqq \sup_{t\ge 0}\|\wti{B}_{f,t,m}\varphi(y)\|_\calG $. 
 In particular, 
 \begin{align}\label{eq:max-op-fully-0p-valued-1}
 	\|\wti{B}_{f,m}^*\varphi\|_{L^2(Y,dy)}^2
 		& \le \|m_2\|_{L^2(\R_+,\frac{ds}{s})}^2 \,\
 				\int_{Y}   \int_0^\infty \|\Phi[m_1(sf)\varphi](y)\|_\calG ^2 \,\frac{ds}{s} \, dy\, .
 \end{align}
 Interchanging the integrals, the last factor on the right hand side of  
 \eqref{eq:max-op-fully-0p-valued-1} is given by 
 \begin{align*}
 	\int_0^\infty &\int_{Y}  \|\Phi[m_1(sf)\varphi](y)\|_\calG ^2 \, \, dy\, \frac{ds}{s}
 		= 
 			\int_0^\infty \|\Phi[m_1(sf)\varphi]\|_{L^2(Y,\calG)}^2 \, \frac{ds}{s} \\
 		&=  \int_0^\infty \|m_1(sf)\varphi\|_{L^2(X,\calH)}^2 \, \frac{ds}{s} 
 			= \int_{X} \int_0^\infty \big\la 
 					m_1(sf(x))\varphi(x), m_1(sf(x))\varphi(x)  
 				\big\ra_\calH  \, \frac{ds}{s} \, dx \\
 		&= \int_{X} \big\la 
 				\varphi(x), \int_0^\infty  m_1(sf(x))^2\, \frac{ds}{s}
 					\varphi(x)
 			  \big\ra_\calH   \, dx \, .
 \end{align*}
 As functions of the real variable $r\ge 0$ the scaling invariance of the measure $ds/s$ on $\R_+ $ and $m_1(0)=0$ gives 
  $ \int_0^\infty  m_1(sr)^2\, \frac{ds}{s} = \|m_1\|_{L^2(\R_+,\frac{ds}{s})}^2\id_{\{r>0\}}$,  so the spectral theorem implies 
 \begin{align*}
 	\big\la 
 		\varphi(x), \int_0^\infty  m_1(sf(x))^2\, \frac{ds}{s}
 					\varphi(x)
 	\big\ra_\calH 
 		&= 
 		\|m_1\|_{L^2(\R_+,\frac{ds}{s})}^2
 		 \big\la 
 			\varphi(x), \id_{\{f(x)>0\}} \varphi(x)
 		 \big\ra_\calH  \\
 		&\le \|m_1\|_{L^2(\R_+,\frac{ds}{s})}^2
 		 		\|\varphi(x)\|_\calH^2 \, .
 \end{align*}
 Using this in \eqref{eq:max-op-fully-0p-valued-1} shows 
 \begin{align}\label{eq:max-op-fully-0p-valued-2}
	\|\wti{B}_{f,m}^*\varphi\|_{L^2(Y,dy)} 
		\le \|m_1\|_{L^2(\R_+,\frac{ds}{s})} \|m_2\|_{L^2(\R_+,\frac{ds}{s})}^2
 		 		\|\varphi\|_{L^2(X,\calH)} \,. 
 \end{align}
 which proves \eqref{eq:bound-B-fully-operator-valued}, since $\|\wti{B}_{f,g,m}\varphi(y)\|_\calG\le \wti{B}_{f,m}^*\varphi(y)$ for all $y\in Y$.
\end{proof}

\bigskip
\appendix
\setcounter{section}{0}
\renewcommand{\thesection}{\Alph{section}}
\renewcommand{\theequation}{\thesection.\arabic{equation}}
\renewcommand{\thetheorem}{\thesection.\arabic{theorem}}

\section{Induction in dimension}\label{sec:induction-in-dimension}
In this section we prove Theorem \ref{thm:CLR-op-valued}, that is, we prove that the number of negative bound states of $P^2 \otimes \idG + V$ is bounded by  
 \begin{align*}
 	N(P^{2}\otimes\idG+V)
			\le C_{0,d}^{\mathrm{op}}\, 
			\frac{|B_1^d|}{(2\pi)^d} \int_{\R^d} \tr_\calG[ V_-(x)^{\frac{d}{2}}]\, dx
 \end{align*}
 and, moreover, 
 \begin{align*}
 	C_{0,d}^{\mathrm{op}} = \min_{3\leq n \leq d} C_{0,n}^{\mathrm{op}} \leq \min_{3\leq n \leq d} C_n,
 \end{align*} 
 where $C_n$ is given by \eqref{eq:C-gamma} for $\gamma=n$. Here, $V:\R^d\to \calB(\calG)$ is an operator valued potential with positive part $V_+\in L^1_{\text{loc}}(\R^d,\calB(\calG))$ and negative part  $V_-\in L^{d/2}(\R^d, \calS_{d/2}(\calG))$.

In order to do this, we need the following operator-valued 
extension of the well-known Lieb--Thirring bounds for suitable 
moments $\theta$:
\begin{align}\label{eq:LT-op-valued}
	\tr_{L^2(\R^d,\calG)} \big[ P^2 \otimes \idG + V \big]_-^{\theta} \leq L_{\theta, d}^{\mathrm{op}} \int_{\R^d} \tr_{\calG} \big[V_-(x)^{\theta+\frac{d}{2}}\big] \, dx,
\end{align}
where $L_{\theta, d}^{\mathrm{op}} = C_{\theta, d}^{\mathrm{op}} \, L_{\theta, d}^{\mathrm{cl}}$ with the classical Lieb--Thirring constant
\begin{align}
	L_{\theta, d}^{\mathrm{cl}} = \int_{\R^d} (1-\eta^2)_+^{\theta}\, \frac{d\eta}{(2\pi)^d}.
\end{align}
It is important that the constant $L_{\theta, d}^{\mathrm{op}}$, respectively, $C_{\theta, d}^{\mathrm{op}}$ does not depend on the auxiliary Hilbert space $\calG$.

The bound \eqref{eq:LT-op-valued} was first proven in the seminal work of Laptev and Weidl \cite{LaWe} for all dimensions $d\in \N$ and moments 
$\theta \geq \frac{3}{2}$, moreover, they showed  
$C_{\theta, d}^{\mathrm{op}}=1$ in this case. This was later simplified in \cite{BeLo}. For moments $\theta\ge \frac{1}{2}$ and again all dimensions $d\in\N$ the bound \eqref{eq:LT-op-valued} was shown to hold in \cite{HLW}, moreover, 
$C_{\theta, d}^{\mathrm{op}}\le 2$ for $\frac{1}{2}\le \theta<\frac{3}{2}$, see also \cite{DLL} and, recently, \cite{FHJN} for improvements when $\theta=1$. The limiting case $\theta=0$, that is, the operator--valued version of the CLR bound was then proven in \cite{Hu}, with improvements on the constant later in \cite{FLS}. 

The possibility that a bound of the form \ref{eq:LT-op-valued} allows to strip off one dimension in the Lieb--Thirring bounds was crucially used in Laptev--Weidl \cite{LaWe}, see also \cite{La}. The possibility of stripping off more than one dimension was realized in \cite{Hu}.

In the short proof below, which we give for the convenience of the reader, we follow the discussion in \cite{Hu}. 

\begin{lemma}\label{lem:submult} For $n\leq d$ we have 
\begin{align*}
	C_{\theta,d}^{\mathrm{op}} \leq C_{\theta,n}^{\mathrm{op}} C_{\theta+\frac{n}{2},d-n}^{\mathrm{op}}.
\end{align*}
In particular, for $d\geq 3$,
\begin{align*}
	C_{0,d}^{\mathrm{op}} \leq C_{0,n}^{\mathrm{op}} \quad \text{for all } 3 \leq n \leq d.
\end{align*} 
\end{lemma}

\begin{proof}
For $n\leq d$ we factor $\R^d = \R^n \times \R^{d-n}$, that is, $x = (x_<, x_>) \in \R^n \times \R^{d-n}$, and split the the kinetic energy as $P^2 = P^2_< + P^2_>$, more precisely,
\begin{align*}
P^2 = P_<^2 \otimes \id_{L^2(\R^{d-n})} + \id_{L^2(\R^n)} \otimes P_>^2.
\end{align*}
Moreover, observe that 
\begin{align*}
	L^2(\R^d, \calG) = L^2(\R^d) \otimes \calG = L^2(\R^n) \otimes L^2(\R^{d-n}) \otimes \calG = L^2(\R^n, L^2(\R^{d-n}\otimes \calG)).
\end{align*}
As quadratic forms on $L^2(\R^d, \calG)$, we then have 
\begin{align}\label{eq:induction-splitting}
\begin{split}
	P^2 \otimes \idG + V(x) &= P_<^2 \otimes \id_{L^2(\R^{d-n})} \otimes \idG + \id_{L^2(\R^n)} \otimes P_>^2 \otimes \idG + V(x_<, x_>) \\
	&\geq P_<^2 \otimes \id_{L^2(\R^{d-n},\calG)} -  W(x_<)
\end{split}
\end{align}
with the operator-valued potential $W(x_<)= \big(P_>^2 \otimes \idG + V(x_<, \cdot) \big)_- : L^2(\R^{d-n},\calG) \to L^2(\R^{d-n},\calG)$. Note that $W(x_<)$ is the negative part of a Schr\"odinger operator in $d-n$ dimensions where one freezes the $x_<$ coordinate in the potential. Inequality \eqref{eq:LT-op-valued} can therefore be applied and yields
\begin{align*}
	\tr_{L^2(\R^{d-n}, \calG)} W(x_<)^{\theta+\frac{n}{2}} &= \tr_{L^2(\R^{d-n}, \calG)} \big(P_>^2 \otimes \idG + V(x_<, 
	\cdot) \big)_-^{\theta+\frac{n}{2}} \\
	&\leq L_{\theta+\frac{n}{2},d-n}^{\mathrm{op}} \int_{\R^{d-n}} \tr_{\calG} V_- (x_<, x_>)^{\theta+\frac{d}{2}} \,dx_<.
\end{align*}
Since by assumption $\int_{\R^{d}} \tr_{\calG} V_- (x)^{\theta+\frac{d}{2}} \, dx <\infty$, the Fubini--Tonelli theorem shows that $W(x_<)$ is compact (even in the von Neumann--Schatten ideal $\calS_{\theta+\frac{n}{2}}(L^2(\R^{d-n},\calG))$) for almost all $x_<\in\R^n$.
Taking traces in inequality \eqref{eq:induction-splitting} gives the estimate
\begin{align*}
	\tr_{L^2(\R^d,\calG)} \big(P^2 \otimes \idG + V \big)_- ^{\theta} 
	&\leq \tr_{L^2(\R^n,L^2(\R^{d-n},\calG))} \big(P_<^2 \otimes \id_{L^2(\R^d-n,\calG)} - W \big)_- ^{\theta} \\
	&\leq L_{\theta, n}^{\mathrm{op}} \int_{\R^n} \tr_{L^2(\R^{d-n},\calG)} W(x_<)^{\theta+\frac{n}{2}}\, dx_< \\
	&\leq L_{\theta, n}^{\mathrm{op}} L_{\theta+\frac{n}{2}, d-n}^{\mathrm{op}} \int_{\R^d} \tr_{\calG} V_-(x)^{\theta+\frac{d}{2}}\, dx
\end{align*}
where we also used the operator-valued Lieb-Thirring inequality \eqref{eq:LT-op-valued} and combined the integrals using the Fubini--Tonelli theorem. 
It follows that 
\begin{align}\label{eq:induction-L}
	L_{\theta,d}^{\mathrm{op}}  \leq L_{\theta, n}^{\mathrm{op}} L_{\theta+\frac{n}{2}, d-n}^{\mathrm{op}}.
\end{align}
A short calculation, see below, shows 
\begin{align}\label{eq:induction-L-classical}
L_{\theta,d}^{\mathrm{cl}} = L_{\theta, n}^{\mathrm{cl}} L_{\theta+\frac{n}{2}, d-n}^{\mathrm{cl}},
\end{align}
so \eqref{eq:induction-L} and the definition of $C_{\theta, d}^{\mathrm{op}}$ imply the sub-multiplicativity
\begin{align*}
		C_{\theta,d}^{\mathrm{op}} \leq C_{\theta,n}^{\mathrm{op}} C_{\theta+\frac{n}{2},d-n}^{\mathrm{op}}.
\end{align*}
which proves is the first claim of Lemma \ref{lem:submult}. 
In particular, for $\theta=0$ and $3\leq n \leq d-1$, we get
\begin{align*}
	C_{0,d}^{\mathrm{op}} \leq C_{0,n}^{\mathrm{op}} C_{\frac{n}{2},d-n}^{\mathrm{op}} = C_{0,n}^{\mathrm{op}}
\end{align*}
since Laptev--Weidl \cite{LaWe} showed $C_{\theta,m}^{\mathrm{op}} = 1$ if $m\in\N$ and $\theta \geq \frac{3}{2}$. This proves the second claim in Lemma \ref{lem:submult}.

It remains to show \eqref{eq:induction-L-classical}, which follows from the definition of the classical Lieb--Thirring constant and the Fubini--Tonelli Theorem:
\begin{align*}
	L_{\theta,d}^{\mathrm{cl}} &= \int_{\R^d} (1-\eta^2)_+^{\theta} \frac{d\eta}{(2\pi)^d} = \iint_{\R^n \times \R^{d-n}} (1-\eta_<^2-\eta_>^2)_+^{\theta} \frac{d\eta_< \, d\eta_>}{(2\pi)^n (2\pi)^{d-n}} \\
	&= \int_{\R^{d-n}} \int_{\R^n} (1-\eta_>)^{\theta+\frac{n}{2}} (1-\xi^2)_+^{\theta} \frac{d\xi}{(2\pi)^n} \frac{d\eta_>}{(2\pi)^{d-n}} = L_{\theta,n}^{\mathrm{cl}} L_{\theta+\frac{n}{2},d-n}^{\mathrm{cl}}
\end{align*}
The third equality follows from a straightforward scaling argument. 
\end{proof}

\begin{proof}[Proof of Theorem \ref{thm:CLR-op-valued}] 
Lemma \ref{lem:submult} shows that 
\begin{align*}
	C_{0,d}^{\mathrm{op}} \leq \min_{3\leq n\leq d} C_{0,n}^{\mathrm{op}}
\end{align*}
	and the reverse inequality clearly holds. 
	Moreover, the case $\alpha=1$ in Theorem \ref{thm:poly-operator-valued} shows the bound 
	\begin{align*}
		C_{0,n}^{\mathrm{op}} \leq C_n
	\end{align*}
	with the constant $C_{\gamma=n}$ from \eqref{eq:C-gamma}.
\end{proof}


\section{Auxiliary bounds for the operator-valued case}
 In this appendix we gather three results, which we needed for extending our method from the scalar case to the operator-valued case. This results are probably well-known to specialist, we give short proves for the convenience of the reader. 
 
 First we consider operators of the form $A^* A$ and $A A^*$ for some bounded operator $A:\calH\to \calG$, where $\calH, \calG$ 
 are two auxiliary (separable) Hilbert spaces. Let 
 $N(A)= \{ f\in \calH:\, Af=0\}\subset \calH$ be the null space of 
 $A$, $N(A^*)= \{ g\in \calG:\, A^*g=0\}\subset \calG$ the null space 
 of the adjoint $A^*:\calG\to\calH$,  and 
 $N(A)^\perp \coloneqq \{g\in\calG:\, \la g, Af\ra_\calG=0\}\subset \calG$, 
 respectively  $N(A^*)^\perp \coloneqq \{f\in\calH:\, \la f, A^*g\ra_\calH=0\}\subset\calH$,
 the orthogonal complement of $N(A)$ in $\calG$, respectively $N(A^*)$ in $\calH$. 

\begin{lemma}\label{lem:unitary-equivalence}
	Let $\calH,\calG$ be Hilbert spaces and $A:\calH\to \calG$ be 
	a bounded operator. Then $A^*A\big|_{N(A^*)^\perp}$ is unitarily equivalent to 
	$A A^*\big|_{N(A)^\perp}$. In particular, if $A:\calH\to\calG$ is compact, then its non-zero singular values, including multiplicities,  are the same as the non-zero singular values of $A^*:\calG\to\calH$. 
\end{lemma}
\begin{remark}
	In Theorem 3 in \cite{Deift} a stronger result, which allows 
	for unbounded operators is proven, we need it only for bounded 
	operators $A:\calH\to \calG$.
\end{remark}
\begin{proof}
 The polar decomposition, e.g., Theorem VI.10 in \cite{RS1}, of a 
 bounded operator easily extends to a two Hilbert space situation: 
 For a bounded operator $A:\calH\to\calG$  there exists a partial isometry 
 $U:\calH\to\calG$ with $N(U)= N(A)$ and range 
 $\mathrm{Ran}(U)=\ol{\mathrm{Ran}(A)}$, and a symmetric operator $|A|$ with 
 $|A|^2= A^* A$ such that $A=U|A|$. 
 \smallskip
 
 Moreover, $U:\ol{\mathrm{Ran}(A^*)}=N(A)^\perp\to \ol{\mathrm{Ran}(A)}= N(A^*)^\perp$ 
 is an isometry, and 
 \begin{align*}
 	A A^* = U|A|^2U^*= U A^* A U^* \, ,  	
 \end{align*}
 so $A A^*\vert_{N(A)^\perp}$ is unitarily equivalent to 
 $A^* A\vert_{N(A)^\perp}$.   
 
 Since the singular values of $A$ are the square roots of the 
 eigenvalues of $A^*A$ and the singular values of $A^*$ the square 
 roots of the eigenvalues of $AA^*$, the last claim in Lemma \ref{lem:unitary-equivalence} is evident from 
 the unitary equivalence above.  
\end{proof}

 Given a Hilbert space $\calH$ and a $\sigma$-finite measure space $(X,dx)$ 
 we denote by $L^p(X,\calH)$ the space of measurable functions 
 $f:X\to \calH$ for which 
 \begin{align}
 	\|f\|_{p}\coloneqq \|f\|_{L^p(X,\calH)}
 		\coloneqq
 			\left( \int_{X} \|f(x)\|_{\calH}^p\, dx \right)^{1/p} <\infty\, ,
 \end{align}
 when $1\le p<\infty$, respectively,   
 \begin{align}
 	\|f\|_{\infty}\coloneqq \|f\|_{L^\infty(X,\calH)}
 		\coloneqq \esssup_{x\in X } \|f(x)\|_\calH <\infty\, ,
 \end{align}
 when $p=\infty$. Since $\calH$ is assumed to be separable, 
 Pettis' measurability theorem \cite{Pettis}, see also \cite{DiUhl}, shows that the weak and strong 
 notions of measurability for functions $X\ni x\mapsto f(x)$ 
 coincide.  
 If $\calH=\C$, we simply write $L^p(X,\C)= L^p(X)$. 
 Moreover, we denote by 
 $\calS_2(L^2(X,\calH), L^2(Y,\calG))$, the space of 
 Hilbert--Schmidt operators 
 $H:L^2(X,\calH)\to  L^2(Y,\calG)$ with scalar-product 
 \begin{align}
 	\la H_1, H_2  \ra_{\calS_2} 
 		\coloneqq 
 			\tr_{L^2(X,\calH)}\left[ H_1^* H_2 \right]
 \end{align}
 and associated norm $\|H\|\calS_2\coloneqq \la H, H  \ra_{\calS_2}^{1/2} $
 and by $L^2(Y\times X,\calS_2(\calH,\calG)) $, the 
	$L^2$--space of operator-valued kernels 
	$K: Y\times X\to \calS_2(\calH,\calG)$ with scalar product  
	\begin{align*}
		\la K_1,K_2  \ra_{L^2(Y\times X,\calS_2(\calH,\calG))} 
		&\coloneqq 
			\iint_{Y\times X} \|K(y,x)\|_{\calS_2(\calH,\calG)}^2\, dy dx \\
		&=  \iint_{Y\times X} \|K(y,x)\|_{\calS_2(\calH,\calG)}^2\, dy dx	
	\end{align*}
 The next result extends the well-known one-to-one correspondence 
 of Hilbert--Schmidt operators from $L^2(X)$ to $L^2(Y)$ with 
 kernels in  $ L^2(Y\times X)$ to the operator-valued setting.

\begin{lemma}\label{lem:hilbert-schmidt}
  Let $(X,dx)$ and $(Y,dy)$ be $\sigma$-finite measure spaces and 
  $\calH,\calG$ two auxiliary Hilbert spaces. 
  Then $\calS_2(L^2(X,\calH), L^2(Y,\calG))$ 
  is isomorphic to $L^2(Y\times X,\calS_2(\calH,\calG)) $, that is, for any 
  $H\in \calS_2(L^2(X,\calH), L^2(Y,\calG))$ there exists a unique 
  $K_H\in L^2(Y\times X,\calS_2(\calH,\calG)) $ such that for any 
  $f\in L^2(X,\calH)$  and almost all $y\in Y$  
  \begin{align*}
		Hf(y) = \int_{X} K_H(y,x) f(x)\, dx
  \end{align*}
 and vice versa. Moreover, the Hilbert--Schmidt norm of $H\in \calS_2(L^2(X,\calH), L^2(Y,\calG))$ can be calculated as 
 \begin{align*}
 	\|H\|_{\calS_2}^2 = \iint_{Y\times X} \tr_{\calH}\left[ K_H(y,x)^*K_H(y,x)\right]\, dx dy \, . 
 \end{align*}
\end{lemma}
\begin{proof}
The proof is a modification of the scalar-valued case. We sketch it for the convenience of the reader. 
 Any kernel $K\in L^2(Y\times X,\calS_2(\calH,\calG)$ yields a 
 bounded operator $H_K: L^2(X,\calH)\to  L^2(Y,\calG)$ by defining  
 \begin{align*}
 	H_Kf(x)\coloneqq \int_{X} K(y,x) f(x)\, dx . 
 \end{align*}
 Indeed, since  
 \begin{align*}
 	\|H_Kf(y)\|_\calG 
 		& \le \int_{X} \|K(y,x) f(x)\|_\calG\, dx
 		\le \int_{X} \|K(y,x)\|_{\calB(\calH,\calG)} \, \|f(x)\|_\calH\, dx \\
 		&\le \left(\int_{X} \|K(y,x)\|_{\calB(\calH,\calG)}^2 \, dx \right)^{1/2} 
 			\|f\|_{L^2(X,\calH)} \, , 
 \end{align*} 
 by Cauchy--Schwarz, we get 
 \begin{equation}
 \begin{split}\label{eq:L2-HS-norm-bound}
 	\|H_Kf\|_{L^2(Y,\calG)}^2 
 		&= \int_{Y} \|Hf(y)\|_\calG^2\, dy 
 			\le \iint_{Y\times X} \|K(y,x)\|_{\calB(\calH,\calG)}^2 \, dx dy 
 					 \, \|f\|_{L^2(X,\calH)}^2 \\
 		&\le \iint_{Y\times X} \|K(y,x)\|_{\calS_2(\calH,\calG)}^2 \, dx dy 
 					 \, \|f\|_{L^2(X,\calH)}^2
 			= \|K\|_{L^2}^2 \|f\|_{L^2(X,\calH)}^2
 \end{split}
 \end{equation}
 since the Hilbert--Schmidt norm bounds the operator norm. So the 
 map $K\mapsto H_K$ from kernels to Hilbert--Schmidt operators is 
 bounded with $\|H_K\|_{\calS_2}\le \|K\|_{L^2}$ and injective.  
 \smallskip
 
 Given two orthonormal bases $(\alpha_m)_{m\in\N}$ of $\calH$ 
 and $(\beta_m)_{m\in\N}$ of $\calG$, the space 
 $S_2(\calH,\calG)$ has a basis given by 
 the rank-one operators    
 $|\beta_m \ra\la\alpha_n|:\calH\to\calG$, 
 $f\mapsto \beta_m \la\alpha_n,f\ra_\calH$. Furthermore, let 
 $(\varphi_j)_{j\in\N}$ and $(\psi_l)_{l\in\N }$ be bases for $L^2(Y)$ 
 and $L^2(X)$. 
 Then  $(\Psi_{l,n})_{l,n\in\N}$, given by the  $\calH$-valued 
 functions $ X\ni x\mapsto \Psi_{l,n}(x) = \psi_l(x) |\alpha_n\ra$, 
 is a basis for $L^2(X,\calH)= L^2(X)\otimes \calH$ and 
 $(\Phi_{l,m})_{k,m\in\N}$, given by the  $\calG$-valued 
 functions $ Y\ni y\mapsto \Phi_{l,n}(y) = \varphi_k(y) |\beta_m\ra$, 
 is a basis for $L^2(Y,\calG)$. 
  Thus any kernel  $K\in L^2(Y\times X, \calS_2(\calH,\calG))
 	= L^2(Y)\otimes L^2(X)\otimes  \calS_2(\calH,\calG)$ can be 
 written in the form 
 \begin{align*}
 	K(y,x) = \sum_{k,l,m,n \in\N} a_{k,l,m,n} \, \varphi_k(y) \ol{\psi_l(x)} |\beta_m \ra\la\alpha_n|
 \end{align*}
 and a short calculation shows 
 \begin{equation}
 	\|K\|_{L^2}^2 
 		= \iint_{Y\times X} \tr\left[ K(y,x)^*K(y,x) \right]\, dx dy
 			= \sum_{k,l,m,n \in\N} |a_{k,l,m,n}|^2 \, .
 \end{equation}
 Let $R\in\N$ and 
 \begin{align}\label{eq:K-R}
 	K_R(y,x) = \sum_{k,l,m,n =1}^R a_{k,l,m,n} \, \varphi_k(y) \ol{\psi_l(x)} |\beta_m \ra\la\alpha_n|\, ,
 \end{align}
 which is the kernel of the finite rank operator 
 \begin{align}\label{eq:operator-K-R}
 	H_{K_L} = \sum_{k,l,m,n =1}^R a_{k,l,m,n} |\Phi_{k,m}\ra\la\Psi_{l,n}| 
 		= \sum_{k,l,m,n =1}^R a_{k,l,m,n} \Phi_{k,m} \la\Psi_{l,n}, \cdot\ra_{L^2(X,\calH)}
 \end{align}
 Since $\|K-K_R\|_{L^2}\to 0$  the bound \eqref{eq:L2-HS-norm-bound} shows $\|H_{K}- H_{K_R}\|\to 0$ as $R\to\infty$, so any $H_K$ is the limit in the operator norm of finite-rank operators, hence a compact operator. Using the basis 
 $(\Psi_{l,n})_{l,n\in\N}$ to calculate the trace, a straightforward  calculation shows 
  \begin{align*}
 	\tr_{L^2(X,\calH)}\left[ H_{K}^* H_{K} \right] 
 		= \sum_{l,n} \| H_{K}\Psi_{l,n}\|_\calG^2 
 		=\sum_{k,l,m,n\in\N } |a_{k,l,m,n}|^2 
 		= \|K\|_{L^2}^2 
 \end{align*}
 so $H_K\in\calS_2(L^2(X,\calH), L^2(Y,\calG))$ and $\|H_K\|_{\calS_2}= \|K\|_{L^2}$. 
 \smallskip
 
 So far we have shown that the map $K\mapsto H_K$ is an isometry from 
 $\in L^2(Y\times X, \calS_2(\calH,\calG))$ into 
 $\calS(L^2(X,\calH), L^2(Y,\calG))$ so its range is closed. The finite 
 rank operators $F:L^2(X,\calH)\to L^2(Y,\calG)$ are of the form 
 \begin{align*}
 	F= \sum_{r,s\in\N } c_{r,s} | \wti{\Phi}_r\ra\la \wti{\Psi}_s |
 		= \sum_{r,s\in\N } c_{r,s} \wti{\Phi}_r \la \wti{\Psi}_s , \cdot\ra_{L^2(X,\calH)}
 \end{align*}
 with $c_{r,s}\not= 0$ for finitely many $r,s\in\N$ and 
 $\wti{\Phi}_r\in L^2(Y,\calG)$, $\wti{\Psi}_s\in L^2(X,\calH) $. 
 Expanding $\wti{\Psi}_s$ in the basis $(\Phi_{l,n})_{l,n\in\N}$  
 and similarly for $\wti{\Phi}_r$, one sees that finite rank operators 
 of the above form  can be arbitrarily well approximated, 
 in operator norm,  by finite rank operators of the form 
 \eqref{eq:operator-K-R}. Since the finite rank operators are dense 
 in the Hilbert--Schmidt operators, the operators of the 
 form \eqref{eq:operator-K-R} are also dense and hence the range of  
 $K\mapsto H_K$ is all of $\calS(L^2(X,\calH), L^2(Y,\calG))$.  
\end{proof}

The last result concerns an operator-valued version of Dunford's theorem. For this we need some more notation. For background on integration in Banach spaces, we refer to \cite{DiUhl}. 

We denote by $\calB(\calH,\calG)$ the Banach space of bounded operators from 
$\calH$ to $\calG$ equipped with the operator norm. 

We write $L^\infty_s(Y\times X,\calB(\calH,\calG))$ for the space of functions $K:Y\times X \to \calB(\calH,\calG)$ such that 
 \begin{align*}
 	\esssup_{(y,x)\in Y\times X} \|K(y,x)\|_{\calB(\calH,\calG)}<\infty,
 \end{align*}
 and for all $h\in\calH$ the map 
 \begin{align*}
 	Y\times X \ni (y,x) \mapsto K(y,x) h \in \calG
 \end{align*}
 is strongly measurable (with respect to the topology on $\calG$). Since $\calG$ is a separable Hilbert space, Pettis' measurability theorem implies that this the case if and only if it is weakly measurable, i.e., for any $\psi\in\calG$,
 \begin{align*}
 		Y\times X \ni (y,x) \mapsto \langle \psi, K(y,x) h \rangle_{\calG}
 \end{align*}
 is measurable. In this case, for $f\in L^1(X,\calH)$, integrals of the form
 \begin{align}\label{eq:PhiK}
 	\Phi_Kf(y)\coloneqq  \int_{X} K(y,x) f(x) \, dx 
 \end{align}
 are well-defined elements in $\calG$ for almost all $y\in Y$, with 
 \begin{align*}
 	\|\Phi_Kf(y)\|_\calG & =\left\| \int_{X} K(y,x) f(x) \, dx \right\|_{\calG} \leq \int_{X} \|K(y,x) f(x)\|_{\calG} \, dx \\
 	&\leq \esssup_{(y,x)\in Y\times X} \|K(y,x)\|_{\calB(\calH,\calG)} \|f\|_{L^1(X, \calH)}.
 \end{align*} 
 Thus, for $K\in L^\infty_s(Y\times X,\calB(\calH,\calG))$, the map $\Phi_K: L^1(X, \calH) \to L^{\infty}(Y,\calG)$ is bounded with 
 \begin{align*}
 	\|\Phi_K\|_{L^1\to L^{\infty}} \leq \esssup_{(y,x)\in Y\times X} \|K(y,x)\|_{\calB(\calH,\calG)}.
 \end{align*}
 The next Lemma shows that the map $K\mapsto \Phi_K$ is even an isometry. 
 \begin{lemma}\label{lem:kernel-op-valued}
  For any bounded operator $\Phi:L^1(X,\calH)\to L^\infty(Y,\calG)$ there exists a kernel $K_\Phi\in L^\infty_s(Y\times X,\calB(\calH,\calG))$ such that 
  \begin{align*}
  	\Phi f(y) = \int_{X} K_\Phi(y,x)f(x) \, dx
  \end{align*}
 for any $f\in L^1(X,\calH)$ and almost all $y\in Y$. Moreover, 
  \begin{align*}
 	\|\Phi\| = \esssup_{(y,x)\in Y\times X} \|K_\Phi(y,x)\|_{\calB(\calH,\calG)}
  \end{align*}
 \end{lemma}
\begin{proof} If $K\in L^\infty_s(Y\times X,\calB(\calH,\calG))$, 
the discussion above shows that the map $\Phi_K$ defined in \eqref{eq:PhiK} is bounded from $L^1(X,\calH)$ to $L^{\infty}(Y,\calG)$ and 
 \begin{align}\label{eq:dunford-bound-1}
 	\|\Phi_K\|_{L^1\to L^{\infty}} \leq \esssup_{(y,x)\in Y\times X} \|K(y,x)\|_{\calB(\calH,\calG)} \eqqcolon \|K\|_{L^\infty}.
 \end{align}
 
 Conversely, assume that $\Phi$ is a bounded map from $L^1(X,\calH)$ 
 into $L^\infty(Y,\calG)$ and choose orthonormal bases $(\alpha_n)_{n\in\N}$ in $\calH$ and  $(\beta_m)_{m\in\N}$ in $\calG$. Then any function 
 $f\in L^1(X,\calH)$ can be identified with a sequence of 
 functions $f=(f_1, f_2,\ldots)$, where $f_l\in L^1(X)$ and 
 $\|f\|_{L^1(X,\calH)}= \|(\sum_{l\in\N}|f_l|^2)^{1/2}\|_{L^1(X)}$, and similarly for $L^1(Y,\calG)$. 
 So without loss of generality, we can assume that  
 $\calH=\calG=l^2(\N)$, i.e., the bounded operators from 
 $\calH\to \calG$ correspond to infinite matrices which map 
 $l^2(\N)$ boundedly into itself.  
 Finally, let $(e_j)_{j\in\N}$ be the canonical basis 
 of $l^2(\N)$.

For $n\in\N$ and $g_l\in L^1(Y)$, $f_l\in L^1(X)$, $l=1,\ldots, n $, the finite linear combinations\footnote{For the equality $L^1(Y) \otimes L^1(X) = L^1(Y\times X)$ one should be a wee bit more precise about the involved topologies in the tensor products: For a Banach space $E$, the algebraic tensor product 
   $L^1(Y) \otimes_{\mathrm{alg}} E$ is the vector space of finite linear combinations $\sum_{l=1}^n g_l\otimes f_l$, where $g_l\in L^1(Y)$ and $f_l\in E$. 
   One equips this vector space with the norm  
   $\|z\|_\pi\coloneqq\inf\{ \sum_l \|g_l\|_{L^1(Y)} \|f_l\|_E :\, z= \sum_{l} g_l\otimes f_l \}  $. 
   Then for the closure 
   $L^1(Y)\widehat{\otimes} E\coloneqq \overline{L^1(Y) \otimes_{\mathrm{alg}}E}^{\|\cdot\|_\pi}$, called the projective tensor product, one has  
   $L^1(Y)\widehat{\otimes} E = L^1(Y,E)$, see \cite[Proposition III.B.28]{Wo} or \cite[Example VIII.10]{DiUhl}. In particular, one has 
   $L^1(Y)\widehat{\otimes}L^1(X)= L^1(Y, L^1(X))= L^1(Y\times X)$. We will not dwell on 
   this fine point any further ;-) . }  of the form 
 \begin{align*}
 	\sum_{l=1}^n g_l\otimes f_l \in L^1(Y) \otimes L^1(X) = L^1(Y\times X)
 \end{align*}
 are dense in $L^1(Y\times X)$. 
 Now assume that $\Phi:L^1(X,l^2(\N))\to L^\infty(Y,l^2(\N))$ is 
 bounded. For $m,n\in\N$ let  
 \begin{align*}
   S_{m,n}(\sum_{l=1}^n g_l\otimes f_l)
   	\coloneqq   
   		\sum_{l=1}^n \la g_l\otimes e_m, \Phi f_l\otimes e_n\ra
 \end{align*}
 which defines a linear functional on the finite linear combinations 
 and is bounded by $\|S_{m,n}\|\le \|\Phi\|$. Thus it has a continuous 
 extension to all of $L^1(Y\times X)$ and since the dual 
 $L^1(Y\times X)^* = L^\infty(Y\times X)$, there exist measurable 
 functions $K^{m,n}_{\Phi} \in L^\infty(Y\times X)$, $m,n\in\N$, such that 
 \begin{align*}
 	\la g\otimes e_m, \Phi f\otimes e_n\ra = \iint_{Y\times X} K^{m,n}_\Phi(y,x) \ol{g(y)} f(x)\, dx dy \, .
 \end{align*}
 Taking unions of countably many zero sets, we can assume that 
 the kernels $K^{m,n}_\Phi(\cdot,\cdot)$ are well--defined for any $m,n\in\N$, up to 
 a common zero set in $Y\times X$. 
 
 Let $l^2_{\text{fin}}(\N)$ be the set of 
 sequences $\alpha=(\alpha_1,\alpha_2,\ldots)$ with only finitely 
 many $\alpha_j$ non--zero, which is dense in $l^2(\N)$. For 
  $\alpha\in l^2_{\text{fin}}(\N)$ and $(y,x)\in Y\times X$ we define the sequence $K_\phi(y,x)\alpha\in \C^\N$ as 
 \begin{align*}
 	(K_\Phi(y,x) \alpha)_m 
 		\coloneqq \sum_{n\in\N}  K_\Phi^{m,n}\alpha_n, \quad \text{for } m\in\N 
 \end{align*}
 The next step is to show that for almost all $(y,x)\in Y\times X$ one has 
 $K_\phi(y,x)\in \calB(l^2(\N),l^2(\N))$. Since 
 $l^2_{\text{fin}}(\N)$ is dense in $l^2(\N)$ one has 
 \begin{align*}
 	\|K_\Phi(y,x)\|_\calB & = \|K_\Phi(y,x)\|_{\calB(l^2(\N), l^2(\N))} 
 		= \\
 		&= \sup\{ \re\la \beta, K_\Phi(y,x)\alpha \ra|\, \alpha,\beta\in l^2_{\text{fin}}(\N), \|\alpha\|_{l^2}=\|\beta\|_{l^2}=1 \} \\
 		&= \sup\{ \sum_{m,n} \re \left(\ol{ \beta}_m, K_\Phi^{m,n}(y,x)\alpha_n\right) |\, \alpha,\beta\in l^2_{\text{fin}}(\N), \|\alpha\|_{l^2}=\|\beta\|_{l^2}=1 \}\, .
 \end{align*}
 Moreover, let $L^1_{\text{fin}}(X,l^2(\N))$ be the set 
 of functions $f=(f_1,f_2,\ldots)\in L^1(X,l^2(\N))$ with only 
 finitely many nonzero $f_j$, which is dense in $L^1(X,l^2(\N))$, 
 and similarly for $L^1_{\text{fin}}(Y,l^2(\N))$.  
 For any $g\in L^1_{\text{fin}}(Y,l^2(\N))$, 
 $f\in L^1_{\text{fin}}(X,l^2(\N))$, we clearly have from the above  
 \begin{align}\label{eq:dunford-equality}
 \begin{split}
 	\la g,\Phi f \ra 
 		&=  \iint_{Y\times X} 
 				\sum_{m,n} \ol{g_m(y)}\, K^{m,n}_\Phi(y,x)  f_n(x)\, dx dy\, \\
 		&= \iint_{Y\times X} 
 				\la g(y), K_\Phi(y,x)  f(x)\ra_{l^2(\N)}\, dx dy.
 \end{split}
 \end{align}
 and with $A= \{(g,f)\in L^1_{\text{fin}}(Y, l^2(\N))\times L^1_{\text{fin}}(X,l^2(\N))\big|\, \|g\|_{L^1(Y,l^2(\N))}= \|f\|_{L^1(X,l^2(\N))}=1 \}$, which is dense in $L^1(Y, l^2(\N))\times L^1(X,l^2(\N))$, one sees 
 \begin{align*}
 	\esssup_{(y,x)\in Y\times X}  \|&K_\Phi(y,x)\|_{\calB} 
 		 = \sup_{(g,f)\in A }  \iint_{Y\times X} \re\big\la g(y), K_\Phi(y,x)  f(x)\big\ra_{l^2(\N)}  \, dy dx \\
 		&= \sup_{(g,f)\in A}  \re\la g,\Phi f \ra  
 			\le \|\Phi\|_{L^1\to L^\infty} \|g\|_{L^1(Y,l^2{\N})} \|f\|_{L^1(X,l^2{\N})}
 \end{align*}
 Thus the kernel $K_\Phi(y,x)$ maps $l^2(\N)$ boundedly into itself 
 uniformly in $(y,x)\in Y\times X$ and from \eqref{eq:dunford-equality} one also 
 gets $\Phi= \Phi_{K_\Phi}$. In addition, the 
 last bound together with \eqref{eq:dunford-bound-1} shows 
 \begin{align*}
 	\esssup_{(y,x)\in Y\times X}  \|K_\Phi(y,x)\|_{\calB} = \|\Phi\|_{L^1\to L^\infty}
 \end{align*}
 so the map 
 $L^\infty_s(Y\times X, \calB(l^2(N),l^2(\N)))\ni K\mapsto \phi_K\in 
 	\calB(L^1X,l^2(\N), L^\infty(Y,l^2(\N)))$ is an isometry. 
\end{proof}

\section{Numerical results}\label{sec:numerics}
In this section we derive upper bounds on the the constants in 
Theorem \ref{thm:poly} and \ref{thm:poly-operator-valued}, 
in particular, the constant $C_{0,d}$ in the bound for the number of bound states of a non--relativistiv one--particle Schr\"odinger operator  from Corollary \ref{thm:CLR}, 
given in  Table \ref{tb:better}.  

Recall that the best constant in our approach is related to the  minimization problem for  
\begin{align*}
	M_{\gamma} = \inf_{m_1, m_2\in L^2(\R_+, \frac{ds}{s})}\left\{ \left( \|m_1\|_{L^2(\R_+, \frac{ds}{s})} \|m_2\|_{L^2(\R_+, \frac{ds}{s})} \right)^{\gamma-2} \int_{0}^{\infty} \left(1-\frac{(m_1 * m_2)(s)}{s}\right)^2 \, s^{2-\gamma}\,\frac{ds}{s}\right\}.
\end{align*}
The choice of $m_1,m_2$ is quite arbitrary. It is important, however, that one has $m_1*m_2(s)\sim s$ for small $s$, in order to make the integral $\int_{0}^{\infty} \left(1-\frac{(m_1 * m_2)(s)}{s}\right)^2 \, s^{2-\gamma}\,\frac{ds}{s}$ finite. 

We reformulate the above problem by making the ansatz
\begin{align*}
\begin{split}
	m_1(s) = s \int_s^{\infty} \xi(r)\,\frac{dr}{r} \quad ,\qquad
	m_2(s) = s \psi(s),
\end{split}
\end{align*}
where $\xi,\psi: \R_+ \to \R$ are such that $\int_0^{\infty} \xi(r)\,\frac{dr}{r} = \int_{0}^{\infty} \psi(r) \,\frac{dr}{r} =1$. 

Then the convolution of $m_1$ and $m_2$ is given by 
\begin{align*}
	m_1 * m_2 (t) = \int_0^{\infty} m_1(t/s) m_2(s)\,\frac{ds}{s} = t \int_0^{\infty}\int_0^\infty  \xi(r) \psi(s) 1_{\{r>t/s\}} \,\frac{dr}{r} \,\frac{ds}{s} 
\end{align*} 
and a short calculation, taking into account the above normalization of 
$\xi$ and $\psi$,  shows  
\begin{equation}
	\begin{split}	
		&\int_{0}^{\infty} \left(1-\frac{(m_1 * m_2)(t)}{t}\right)^2 \, t^{2-\gamma}\,\frac{dt}{t} = \int_0^{\infty} \left( \iint_{0}^{\infty} 1_{\{r\leq t/s\}} \xi(r) \psi(s)\,\frac{dr}{r} \,\frac{ds}{s} \right)^2  t^{2-\gamma} \,\frac{dt}{t} \\
		&\quad= \underbrace{\frac{1}{\gamma-2} \iiiint_{0}^{\infty} \xi(r_1) \xi(r_2) \psi(s_1) \psi(s_2) \, \max\{ r_1 s_1, r_2 s_2 \}^{2-\gamma} \,\frac{dr_1}{r_1}\,\frac{dr_2}{r_2} \,\frac{ds_1}{s_1} \,\frac{ds_2}{s_2} }_{\eqqcolon I_{\gamma}[\xi, \psi]}.
	\end{split}
\end{equation}
The $L^2$-norms of $m_1, m_2$ can be expressed in terms of $\xi$ and $\psi$ by 
\begin{align*}
	\int_{0}^{\infty} m_1(s)^2 \,\frac{ds}{s} = \int_{0}^{\infty} \left( s \int_{0}^{\infty} \xi(r)\,\frac{dr}{r}  \right)^2 \,\frac{ds}{s} 
	&= \frac{1}{2} \iint_0^{\infty} \xi(r_1) \xi(r_2) \, \min\{r_1, r_2\}^2 \,\frac{dr_1}{r_1}\,\frac{dr_2}{r_2}\\
\intertext{and} 
	\int_{0}^{\infty} m_2(s)^2 \,\frac{ds}{s} 
	&= \int_{0}^{\infty} s^2 \psi(s)^2 \,\frac{ds}{s}. 
\end{align*}
Thus, an upper bound on $M_{\gamma}$ can be obtained by minimizing the functional 
\begin{align}\label{eq:new-min-1}
	\left(\int_{0}^{\infty} s^2 \psi(s)^2 \,\frac{ds}{s}\right)^{\frac{\gamma-2}{2}} \left(\frac{1}{2} \iint_0^{\infty} \xi(r_1) \xi(r_2) \, \min\{r_1, r_2\}^2 \,\frac{dr_1}{r_1}\,\frac{dr_2}{r_2} \right)^{\frac{\gamma-2}{2}} I_{\gamma}[\xi,\psi]
\end{align} 
over all functions $\psi, \xi \in L^1(\R_+,\frac{ds}{s})$ satisfying the constraint 
\begin{align} \label{eq:new-min-2}
\int_0^{\infty} \xi(r)\,\frac{dr}{r} = \int_{0}^{\infty} \psi(r) \,\frac{dr}{r} =1.
	\end{align}
Finding the minimizer, even finding that a minimizer exists for the new minimization problem given by \eqref{eq:new-min-1} and \eqref{eq:new-min-2}, is a very challenging problem, as challenging as for the original minimization problem. 
However, to get a reasonable upper bound on the minimal value, it suffices to take suitable trial functions. 
To get the constants given in Table \ref{tb:better}, in our calculations, which where done with \texttt{Mathematica}, we used the following family of trial functions
\begin{equation}
  \begin{split}\label{eq:trial-functions}
	\xi(s) &= \frac{\alpha^{p}}{\Gamma(p)} s^{-\alpha} (\log s)^{p-1} 1_{\{ s>1\}}, \\
	\psi(s) &= \frac{\beta^{q}}{\Gamma(q)} s^{-\beta} (\log s)^{q-1} 1_{\{ s>1\}},	
  \end{split}
\end{equation}
with parameters $\alpha, p, \beta, q >0$, i.e., Gamma distributions on $\R_+$. 

The normalization condition is easily verified. 
For integer $p,q \geq 1$, the calculation of $I[\xi, \psi]$ can be reduced to calculating the integral
\begin{align*}
	J(\alpha_1, \alpha_2, \beta_1, \beta_2) = \iiiint_1^{\infty} r_1^{-\alpha_1} r_2^{-\alpha_2} s_1^{-\beta_1} s_2^{-\beta_2} \, \max\{ r_1 s_1, r_2 s_2 \}^{2-\gamma} \,\frac{dr_1}{r_1}\,\frac{dr_2}{r_2} \,\frac{ds_1}{s_1} \,\frac{ds_2}{s_2},
\end{align*}
as from $J$ we can get $I[\xi, \psi]$ by taking derivatives,
\begin{align*}
	I[\xi,\psi] = \left.\frac{1}{\gamma-2} \frac{\alpha^{2p}\beta^{2q}}{\Gamma(p)^2 \Gamma(q)^2} \left(\partial_{\alpha_1}\partial_{\alpha_2}\right)^{p-1} \left(\partial_{\beta_1}\partial_{\beta_2}\right)^{q-1}  J(\alpha_1, \alpha_2,\beta_1, \beta_2) \right|_{\substack{\alpha_1=\alpha_2=\alpha \\ \beta_1=\beta_2=\beta}}.
\end{align*}
Similarly, the ``$L^2$-norm integrals'' are given by 
\begin{align*}
	\int_0^{\infty} s^2 \psi(s)^2 \,\frac{ds}{s} = \frac{\beta^{2q}}{2^{2q-1} (\beta-1)^{2q-1}} \frac{\Gamma(2q-1)}{\Gamma(q)^2}
\end{align*}
for $q\in\N$ and $\beta >1$, as well as
\begin{align*}
	\left.\frac{1}{2} \iint_0^{\infty} \xi(r_1) \xi(r_2) \, \min\{r_1, r_2\}^2 \,\frac{dr_1}{r_1}\,\frac{dr_2}{r_2} = \frac{1}{2} \frac{\alpha^{2p}}{\Gamma(p)^2} \left(\partial_{\alpha_1}\partial_{\alpha_2} \right)^{p-1} K(\alpha_1, \alpha_2)\right|_{\alpha_1 = \alpha_2 = \alpha},
\end{align*}
where 
\begin{align*}
	K(\alpha_1, \alpha_2) = \iint_1^{\infty} r_1^{-\alpha_1} r_2^{-\alpha_2} \, \min\{ r_1, r_2\}^2 \,\frac{dr_1}{r_1}\,\frac{dr_2}{r_2} = \frac{\alpha_1 + \alpha_2}{\alpha_1 \alpha_2 (\alpha_1 + \alpha_2 - 2)}
\end{align*}
for $\alpha>1$, $p\in\N$. 

In our numerical calculations with \texttt{Mathematica}, 
we made the choice $p=2, q=3 $ for dimensions $d=3,4$ and 
optimized in the parameters $\alpha, \beta >1$, 
while for dimensions $d \geq 5$ the values were obtained with 
$, p=3, q=2$ and minimization in $\alpha, \beta >1$. 
Move specifically, we got the values in Table \ref{tb:better} by the following choice of parameters.  

  \begin{table}[ht!]
  \begin{tabular}{| l | c | c | c| c| c|}
    \hline
    \multirow{2}{*}{$d$} & $C_{0,d}$  & \multicolumn{4}{c|}{Value of parameters in \eqref{eq:trial-functions}} \\
        &   & $p$ & $\alpha$ &  $q$  & $\beta$         \\ \hline
    3   & 7.55151 & 2  & 2.93254 & 3 & 2.49795 \\ 
    4   & 6.32791 & 2 & 3.69214 & 3 & 2.78716 \\  
    5   & {5.95405} & 3 & 5.46494 & 2 & 2.39433 \\   
    6   & {5.77058} & 3 & 6.41334 & 2 & 2.51583 \\ 
    7   & {5.67647} & 3 & 7.35963 & 2 & 2.61721 \\ 
    8   & {5.63198} & 3 & 8.30512 & 2 & 2.70368 \\ 
    9   & {5.62080} & 3 & 9.25042 & 2 & 2.77865 \\ \hline
  \end{tabular}
  \vskip1ex
  \caption{}
  \label{tb:better-numeric}
  \end{table}

%


\bigskip
\noindent
\section*{Acknowledgements}  We gratefully acknowledge financial support by the Deutsche Forschungsgemeinschaft (DFG) through CRC 1173. Dirk Hundertmark also thanks the Alfried Krupp von Bohlen und Halbach Foundation for financial support. 

We would also like to thank Mathematisches Forschungsinstitut Oberwolfach (MFO) and the Centre International de Rencontres Math\'ematiques (CIRM Luminy) for their research in pairs programmes, where part of this work was conceived.

\bigskip

\providecommand{\mr}[1]{\href{http://www.ams.org/mathscinet-getitem?mr=#1}{MR~#1}}
\providecommand{\zbl}[1]{\href{https://zbmath.org/?q=an:#1}{Zbl~#1}}
\providecommand{\arxiv}[1]{\href{https://arxiv.org/abs/#1}{arXiv:#1}}
\providecommand{\doi}[2]{\href{https://doi.org/#1}{#2}}

\normalsize
\bigskip
\end{document}